\documentclass[journal]{IEEEtran}                                                          % paper

\IEEEoverridecommandlockouts                              % This command is only
                           
\newcommand{\ignore}[1]{}
\usepackage{amsthm}

\usepackage{psfrag}

\usepackage{amsmath,amssymb}

\usepackage{amsfonts}
\usepackage{algcompatible}
\usepackage[linesnumbered,vlined,ruled,boxed,commentsnumbered]{algorithm2e}
\usepackage{cite,color,comment,xspace}
\usepackage{times}
\usepackage{enumerate}
\usepackage{graphicx}
\usepackage{subfigure}
\usepackage{multirow}
\usepackage{epsfig}
\usepackage{nomencl}
\usepackage{textcomp}
\usepackage{soul}
\makenomenclature

\graphicspath{{figs/}}

\newtheorem{theorem}{Theorem}
\newtheorem{lemma}[theorem]{Lemma}
\newtheorem{remark}[theorem]{Remark}
\newtheorem{proposition}[theorem]{Proposition}

\newtheorem{definition}[theorem]{Definition}

\newtheorem{example}[theorem]{Example}
\newtheorem{rul}{Reduction Rule}
\newtheorem{condi}{Condition}
\newcommand{\red}{\color{red}}
\newcommand{\blue}{\color{blue}}

\def\b#1{\mathchoice{\hbox{\boldmath $\displaystyle #1$}}
        {\hbox{\boldmath $\textstyle #1$}}
        {\hbox{\boldmath $\scriptstyle #1$}}
        {\hbox{\boldmath $\scriptscriptstyle #1$}}}
\newcommand{\preset}[1]{\ensuremath{\,\!^\bullet{#1}}}
\newcommand{\postset}[1]{\ensuremath{{#1}^\bullet}}
\newcommand{\N}{{\ensuremath{\mathcal{N}}}}    % a net
\newcommand{\nat}{\mathbb{N}}

\def\BibTeX{{\rm B\kern-.05em{\sc i\kern-.025em b}\kern-.08em
    T\kern-.1667em\lower.7ex\hbox{E}\kern-.125emX}}

\begin{document}

\title{On Liveness Enforcement of Distributed Petri Net Systems}

%This work was supported by the Spanish Ministry of Science and Innovation (MICINN)  [TIN2011-27479-C04-01].

\author{Daniel Clavel, Cristian Mahulea and Manuel Silva \thanks{The authors are with
Arag\'on Institute of Engineering Research (I3A), University of Zaragoza, Spain, emails: \{clavel, cmahulea, silva\}@unizar.es.}}

%\maketitle \thispagestyle{empty} \pagestyle{empty}

\maketitle
\thispagestyle{empty}

\begin{abstract}
This paper considers the liveness enforcement problem in a class of Petri nets (PNs) modeling distributed systems called \emph{Synchronized Sequential Processes} (SSP). This class of PNs is defined as a set of mono-marked state machines (sequential machines, called also agents) cooperating in a distributed way through buffers.  These buffers could model intermediate products in a production system or information channel in a healthcare system but they should be \emph{destination private} to an agent. The designed controller for liveness enforcement should preserve this important property characteristic to the distributed systems. The approach in this paper is based on the construction of a \emph{control PN} that is an abstraction of the relations of the T-semiflows and buffers. The control PN will evolve in parallel with the system, avoiding the firing of transitions that may lead the system to livelock. An algorithm to compute this control PN is presented. Moreover, in order to ensure the liveness of control PN, another algorithm is proposed allowing the firing of local T-semiflow in the correct proportion. Finally, an algorithm for guiding the system evolution is also proposed. 
\end{abstract}

\section{Introduction}
One of the most important logical property of discrete systems is \emph{liveness}. A system is \emph{live} for a given initial state (marking) if it will never reach a (partially) blocking state. In Petri net (PN) literature, liveness analysis has been extensively studied and there exist many results. Furthermore, \emph{structural liveness} (a PN is structurally live if there exists an initial marking that is making the net system live, see for example \cite{ARMura89,ArSiTeCo98,ezpeleta,BOIoAn06}) can be studied by the so called rank theorem \cite{SilvaDSSP,ArSiTeCo98}. However, if the system is not live or not structurally live, constructing a controller in order to force these properties it is a very challenging problem. For some particular classes of Petri nets, there exists many results, as for example in Resource Allocation Systems (RAS) \cite{ezpeleta,park2001deadlock,Colom:2003,li2004elementary,cano2012,7870672}. RAS systems are modular PNs composed by different processes \emph{competing} for shared resources.

In this paper, we consider a different class of systems than the one considered in RAS. In particular, this paper considers modular systems composed also by different processes, called \emph{agents}, but the difference is that here, they are \emph{cooperating} in a distributed way. This cooperation is realized through a set of \emph{buffers} to/from which the agents consume/produce (partial) products. This new class is defined in Def. \ref{def:ssp} and it is called \emph{Synchronized Sequential Processes} (SSP). In order to allow distribution, there exist one important constraint in the assignment of buffers, namely \emph{tokens from a given buffer can only be consumed by a particular agent}, i.e., are destination private but not output private. 

Since distributed systems are considered in this paper, agents are in general geographically distributed. Therefore it is quite natural to assume that any input buffer is in the same location as the corresponding agent and only that agent can consume intermediate products from the buffer.
If the system is modeling a healthcare system, buffers could model information channels containing information to be used only by a clinical pathway \cite{ARBeMaAl19} because of privacy, for example. One advantage of keeping the system distributed is to have more computational attractive analysis approaches since some properties of the system can be studied using the local perspectives due to its modular structure. 

SSP are somehow derived from the well known class of Deterministically Synchronized Sequential Processes (DSSP)\cite{SilvaDSSP}. However, there are some differences that are explained after Def. \ref{def:ssp}. One of the two main differences is that in SSP, buffers can constraints the internal choices, fact that is not allowed in DSSP. However, if a DSSP is not live, in order to force liveness, it is necessary to control the conflicts and by adding some controllers and the net system will not be anymore DSSP. As it is well-known, structurally live and structurally bounded PNs should be consistent and
conservative, two structural properties that assumed all along this work (they can be checked in polynomial time using linear programming \cite{ArSiTeCo98}).

The most important aspect of the purposed liveness enforcement controller in this paper is to keep the distributed property of the system, hence to keep the buffers (included the new \emph{control} buffers) destination private. For this reason, it is not possible to use the results from RAS, a topic discussed in Sec. \ref{sec:dssp}. 

\begin{figure}[h] 
\begin{center}
\centering
\psfrag{b5}{$b_{5}$}\psfrag{b6}{$b_{6}$}
\psfrag{b7}{$b_{7}$}\psfrag{b8}{$b_{8}$}
\psfrag{b1}{$b_{1}$}\psfrag{b2}{$b_{2}$}
\psfrag{b3}{$b_{3}$}\psfrag{b4}{$b_{4}$}
\psfrag{Car A}{$Car A$}\psfrag{Car B}{$Car B$}
 \includegraphics[width=.8\columnwidth]{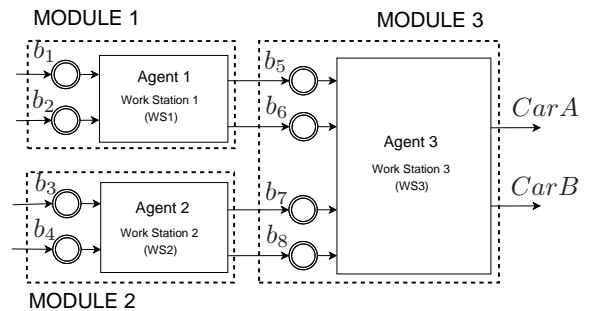} 
\caption{\small Motivation example: a distributed production system composed by three work stations and eight buffers} \label{fig:PracExam}
\end{center}
\end{figure}

As a motivation example, let us consider the production system which constructs cars represented in Fig. \ref{fig:PracExam}. Two different types of cars can be manufactured (model A or model B). The system is composed by 3 work stations denoted as WS1, WS2, and WS3. The WS1 consumes raw material from input buffers $b_1$ (model A) or $b_2$ (model B) and produces intermediate products (engines) to buffers $b_5$ (model A) or $b_6$ (model B). Similarly, WS2 consumes raw material from buffers $b_3$ (model A) or $b_4$ (model B) and produces intermediate products (windshield) to buffers $b_7$ (model A) or $b_8$ (model B). Finally, the WS3 manufactures cars of type A or type B by creating the corresponding bodywork and assembling the intermediate products. Engine from $b_5$ and windshield from $b_7$ are used to obtain cars of type A while engine from $b_6$ and windshield from $b_8$ are used to produce cars of type B. The system can be distributed in 3 different modules (shown in Fig. \ref{fig:PracExam}), one for each agent. If these three agents are geographical distributed, then all input buffers of agents should be private and located at the same position. 

The intermediate products of the input buffers can be assigned to the required activity exactly in the moment when they are needed but also they can be \emph{pre-assigned}. In the case of car producing, first the bodywork is created by agent 3, then an engine from $b_5$ or $b_6$ should be assembled in a second step and finally, in the third step a windshield from $b_7$ or $b_8$ is added. Notice that the engine is assigned to the process in the second step while the windshield in the third step. However, one can \emph{reserve} (or pre-assign) them when the process is started (when the bodywork is started to be produced) and in this way if the process of producing a car is started will finish for sure. Using this pre-asigment approach, in \cite{clavelCDC} a method to ensure that the blocking situations (including partial blocking) will not appear has been proposed. Section \ref{sec:dssp} discusses this approach and also the one based on controlling bad siphons.

The control policy proposed in this paper is based on two main ideas: (i) a local T-semiflow can start firing only if all its input buffers have enough tokens to fire all its transitions and (ii) when the firing of a local T-semilfow starts, all its transitions should be fired before start firing other local T-semiflow, i.e., do not locally interrupt a production task. For the implementation of this control policy, a control PN is defined being the scheduler of the original SSP. Both systems, the original non live SSP and the control PN will evolve synchronously such that the control PN will prohibit the firing of transitions that may lead the SSP to livelock.

In order to enforce liveness of the SSP, it is necessary to ensure that the control PN is live. An algorithm to force the liveness in the control PN is proposed by adding new constraint places. These new places can be seen as new \emph{virtual} buffers with only one output transition. Consequently, the distributiveness of the SSP systems is preserved.

The paper is organized as follows. Sec.~\ref{sec:preliminaries} shows basic concepts and notations. Sec.~\ref{sec:dssp} points out the problems of the techniques based on controlling bad-siphons and that used in~\cite{clavelCDC} and provides some intuition behind the proposed method in this paper. In Sec.~\ref{sec:alg} an algorithm to build the control PN from a SSP structure is given. In Sec.~\ref{sec:live} a methodology to ensure the liveness in the control PN is described while in Sec.~\ref{sec:contr} rules to \emph{guide} the SSP evolution through the control net system are given. Finally, in Sec.~\ref{sec:con}, some conclusions and future works are considered.

\section{Preliminaries}\label{sec:preliminaries}
The reader is assumed to be familiar with Petri nets (see ~\cite{ARMura89,ICSilv93b} for a gentle introduction). The aim of this section is to fix the notation and to recall the required material.

\emph{Nets and Net Systems.}
We denote a Petri Net (PN) as $\N=\langle P, T, \b{Pre}, \b{Post} \rangle$, where $P$ and $T$ are two non-empty and disjoint sets of \emph{places} and \emph{transitions}, and $\b{Pre}, \b{Post} \in \nat^{|P| \times|T|}_{\geq 0}$ are pre and post \emph{incidence matrices}. For instance, $\b{Post}[p,t]=w$ means that there is an \emph{arc} from $t$ to $p$ with \emph{weight} (or multiplicity) $w$. When all arc weights are one, the net is \emph{ordinary}. For pre and postsets we use the conventional dot notation, e.g., $\preset t = \{ p \in P | \b{Pre}[p,t] \neq 0\}$. If $\N'$ is the subnet of $\N$ defined by $P' \subset P$ and $T'  \subset  T$, then $\b{Pre'} =\b{Pre}[P',T']$ and $\b{Post'} =\b{Post} [P', T']$.

A \emph{marking} is a $|P|$ sized, natural valued vector. A Petri Net system is a pair  $\Sigma = \langle \N, \b{m}_0 \rangle$, where $ \b{m}_0$ is the \emph{initial marking}. A transition $t$ is \emph{enabled} at a given marking $\b{m}$ if $\b{m} \geq \b{Pre}[P, t]$; its firing yields a new marking $\b{m}' = \b{m} + \b{C}[P, t]$, where $\b{C}= \b{Post}-\b{Pre}$ is the token-flow matrix of the net. This fact is denoted by $\b{m} \xrightarrow{t} \b{m}'$. An \emph{occurrence sequence} from $\b{m}$ is a sequence of transitions $\sigma= t_1, \ldots, t_k, \ldots$ such that $\b{m} \xrightarrow{t_1} \b{m}_1 \ldots \b{m}_{k-1} \xrightarrow{t_k}\ldots$. The set of all reachable markings, or \emph{reachability set}, from $\b{m}$, is denoted as $RS(\N, \b{m})$. \ignore{The reachability relation is conventionally represented by a \emph{reachability graph} $RG(\N, \b{m})$ where the nodes are the reachable markings and there is an arc labeled $t$ from node $\b{m}'$ to $\b{m}''$ if $\b{m}' \xrightarrow{t} \b{m}''$.}

\emph{State (transition) equation.}
Denoting by $\b \sigma \in \nat^{|T|\times 1}_{\geq 0}$ the \emph{firing count vector} of $\sigma$, where $\b \sigma[t_i]$ is the number of times that $t_i$ appears in $\sigma$. Given $\sigma$ such that $\b{m} \xrightarrow{\sigma} \b{m}'$, then $$\b{m}' = \b{m} + \b{C} \cdot \b \sigma.$$ This is known as the \emph{state transition equation} of $\Sigma$. \ignore{Nevertheless, not necessarily a vector that satisfies the state equation is an actually reachable marking, because the state equation does not check fireability of a sequence with firing count vector $\b{\sigma}$. Such markings vectors are called \emph{spurious markings} ~\cite{ArSiTeCo98}.}

\emph{Place Markings Bounds and Structural Boundedness.}
The marking bound of a place p in a net system $\Sigma$ is defined as: $\b b[p]=max\{\b m[p]\ |\ \b m \in RS(\Sigma)\}$. When this bound is finite, the place is said to be bounded. A net is structurally bounded if every place is bounded for every initial marking.   

\emph{Liveness, Structural Liveness and Deadlocks.}
Liveness is a property related to the potential fireability in all reachable markings. A transition is live if it is potentially fireable in all reachable markings. In other words, a transition is live if it never losses the possibility of firing (i.e., of performing an activity). A transition $t$ is potentially fireable at $\b{m}$ if there exists a firing sequence $\sigma$ leading to a marking $\b{m}'$ in which $t$ is enabled, i.e., $\b{m}[\sigma \rangle \b{m}'\geq \b{Pre}[P,t]$.  A net system is live if all the transitions are live.

A net is structurally live if there exists at least one live initial marking. Non-liveness for arbitrary initial markings reflects a pathology of the net structure: structural non-liveness. In deadlock markings all transitions are dead, so none of them can be fired. A net system is said to be deadlock-free if at least one transition can be fired from any reachable marking. Liveness is a stronger condition than deadlock-freeness.

\emph{Siphons.}
In ordinary PNs, a siphon is a subset of places such that the set of its input transitions is contained in the set of its output transitions: $S \subseteq P$ is a siphon if $\preset {S} \subseteq \postset {S}$. A siphon is minimal if any subset of it is not a siphon. A \emph{bad siphon} is a siphon not containing any trap (set of places that in ordinary PN remain marked for all possible evolution if initially marked).

\ignore{ The following two properties are satisfied in \emph{ordinary} PNs:
\begin {enumerate}
\item If $\b{m}$ is a behavioral deadlock (i.e., dead-marking), in an ordinary net then $S = \{p\ |\ \b{m}[p]=0\}$ is an unmarked (empty) siphon.
\item If a siphon is (or becomes) unmarked, it will remain unmarked for any possible evolution. Therefore all its input and output transitions are dead. So the system is not live (but can be deadlock-free).
\end {enumerate}
}

\emph{T-semiflows, P-semiflow.}
T-semiflows are nonnegative right annullers of $\b C$. So a vector $\b x \gneq 0$ is a T-semiflow if $\b C \cdot \b x =0$. We denote by $||\b x||$ the support of the vector $\b x$ containing all non-null elements: $||\b x|| = \{i | \b x[i] \neq 0\}$. T-semiflow $\b x$ is said to be \emph{minimal} when no T-semiflow $\b{x}'$ exists such that $||\b {x}'|| \subset ||\b x||$. 
The P-semiflows are the nonnegative left annullers of $\b C$.  A vector $\b y \gneq 0$ is a P-semiflow if $\b y \cdot \b C =0$.

From the existence of P- or T-semiflows derives interesting information about the possible behaviors. If a P-semiflow $\b y > \b 0$ exists, $\N$ is \emph{conservative}. Conservativeness ensures structural boundedness. Moreover, $\N$ is \emph{consistent} if a T-semiflow $\b x > \b 0$ exists. A system that is live and bounded must be consistent because a marking repetitive sequence containing all the transitions corresponds to a positive T-semiflow.

\ignore{
\emph{Conflicts and Structural Conflicts.}
A \emph{conflict} is the situation when not all enabled transitions can occur at once. Formally, $t, t^\prime \in T$ are in conflict relation at marking $\b{m}$ if there exist $k, k^\prime \in \nat$ such that $\b{m}  \geq k \cdot \b{Pre}[P,t]$ and $\b{m}  \geq k^\prime \cdot \b{Pre}[P,t^\prime]$, but $\b{m} \ngeq k \cdot \b{Pre}[P,t] + k^\prime \cdot \b{Pre}[P,t^\prime]$. To fulfill the above condition it is necessary that $\preset t \cap \preset t^\prime \neq \emptyset$. When $\b{Pre}[P,t]= \b{Pre}[P,t^\prime] \neq 0$, t and $t^\prime$ are in \emph{equal conflict (EQ) relation}. This means that they are both enabled whenever one is. By defining that a transition is always in EQ with itself, this is an \emph{equivalence} relation on the set of transitions and each equivalence class is an \emph{equal conflict set} denoted, for a given $t$, $EQS(t)$. SEQS is the set of all the equal conflict sets of a given net.}

\emph{Implicit Places.} In general, places impose constraints on the firing of their output transition. When they never do it in isolation, they could be removed without affecting the behaviour of the rest of the system. These places are called implicit. Formally, let $\Sigma = \langle P \cup \{p\}, T, \b{Pre}, \b{Post}, \b{m}_0 \rangle$ be a PN system, the place $p$ is implicit if $\b m \geq \b {Pre}[P,t] \Rightarrow \b m[p] \geq \b {Pre}[p,t]$ for all $t \in \postset p$.

\begin{figure*}[h]
\begin{center}
\centering
\psfrag{t1}{$t_1$}\psfrag{t2}{$t_2$}\psfrag{t3}{$t_3$}\psfrag{t4}{$t_4$}
\psfrag{t5}{$t_5$}\psfrag{t6}{$t_6$}\psfrag{t7}{$t_7$}\psfrag{t8}{$t_8$}
\psfrag{t9}{$t_9$}\psfrag{t10}{$t_{10}$}\psfrag{t11}{$t_{11}$}
\psfrag{t12}{$t_{12}$}\psfrag{t13}{$t_{13}$}\psfrag{t14}{$t_{14}$}
\psfrag{p1}{$p_1$}\psfrag{p2}{$p_2$}\psfrag{p3}{$p_3$}\psfrag{p4}{$p_4$}
\psfrag{p5}{$p_5$}\psfrag{p6}{$p_6$}\psfrag{p7}{$p_7$}
\psfrag{p8}{$p_8$}\psfrag{p9}{$p_9$}
\psfrag{p10}{$p_{10}$}\psfrag{p11}{$p_{11}$}
\psfrag{b1}{$b_{1}$}\psfrag{b2}{$b_{2}$}
\psfrag{b3}{$b_{3}$}\psfrag{b4}{$b_{4}$}
\psfrag{b5}{$b_{5}$}\psfrag{b6}{$b_{6}$}
\psfrag{b7}{$b_{7}$}\psfrag{b8}{$b_{8}$}
\psfrag{N1}{$\N_1$}\psfrag{N2}{$\N_2$}\psfrag{N3}{$\N_3$}\psfrag{pm}{$p_m$}
%\psfrag{engine A}{$engine A$}\psfrag{engine B}{$engine B$}
%\psfrag{windshield A}{$windshield A$}\psfrag{windshield B}{$windshield B$}

   \centering \subfigure[]{\includegraphics[width=0.85\columnwidth]{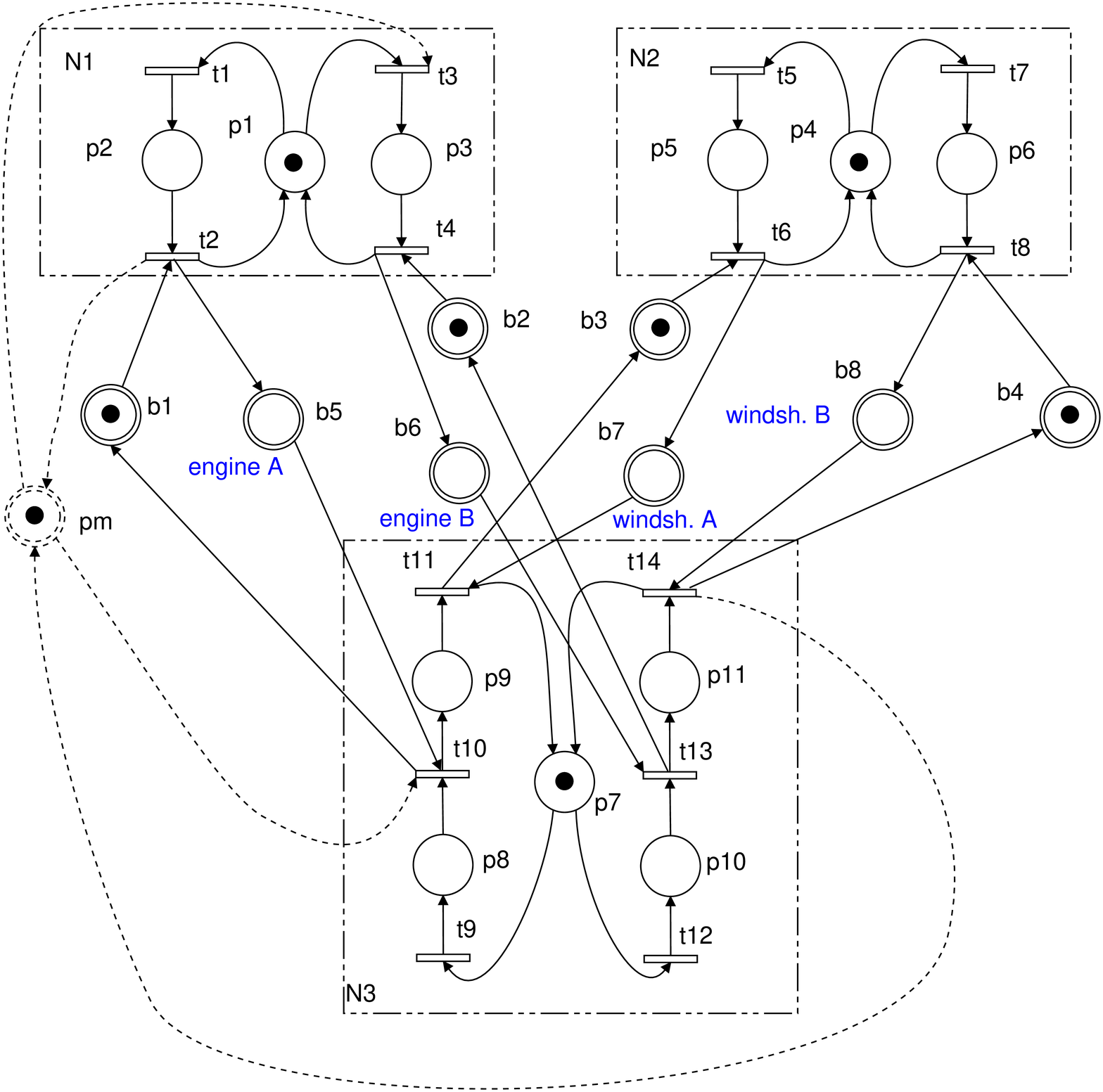}\label{fig:DSSPmot}}\hspace{0.01\textwidth}
	 \centering \subfigure[]{\includegraphics[width=0.89\columnwidth]{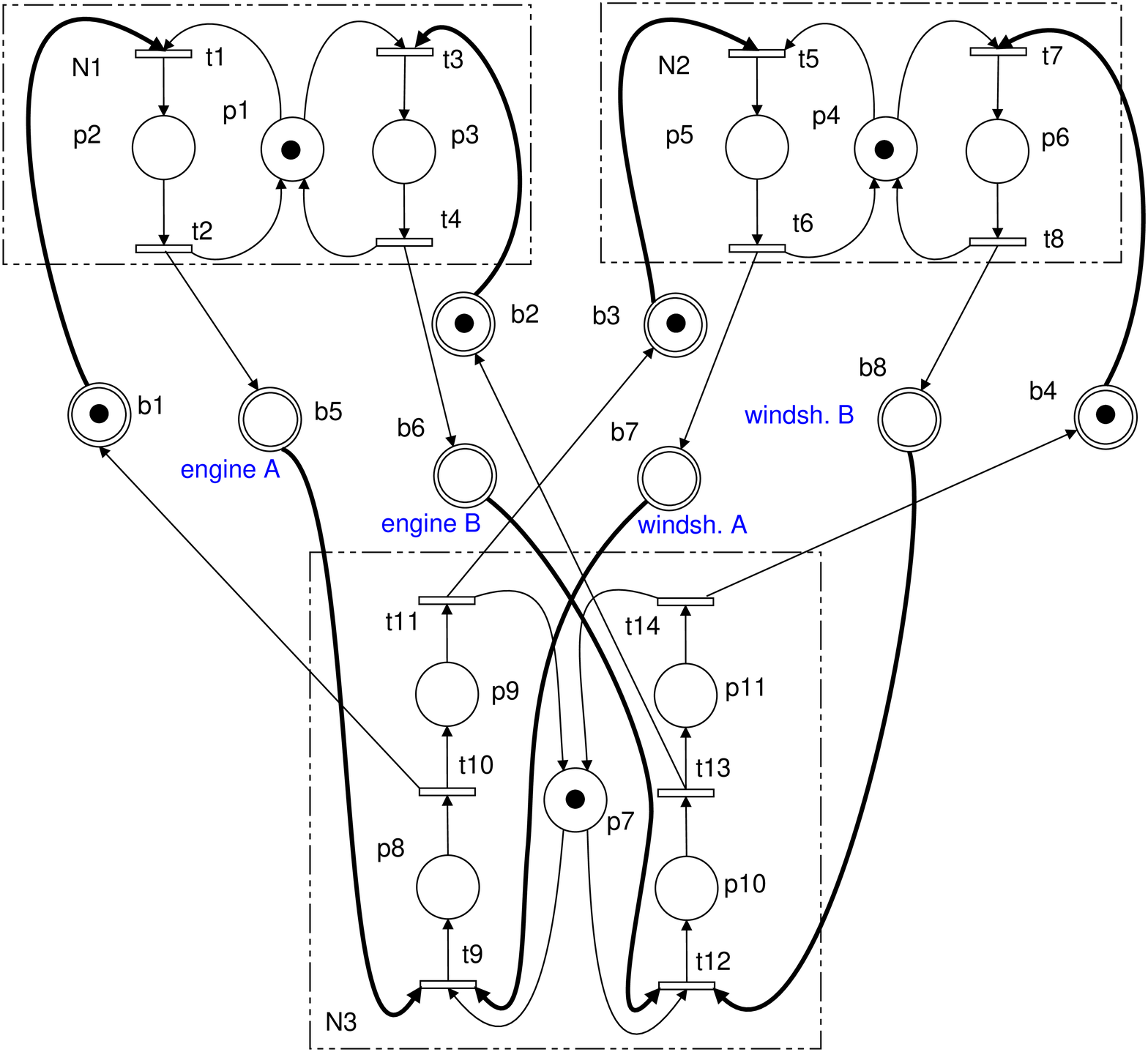}\label{fig:DSSPmotMon}}
\caption{\small Motivation example: (a) The not live SSP net modeling the system in Fig. \ref{fig:PracExam} with he monitor place $p_m$ that prevents the emptiness of the bad siphon $\{p_1,p_2,p_7,p_9,p_{10},p_{11},b_2,b_5\}$; (b) Live net obtained by applied the pre-assignment method in \cite{clavelCDC} to the SSP in Fig.~\ref{fig:first}.a.} \label{fig:first}
\end{center}
\end{figure*}

\emph{PN classes}. If each transition in a ordinary PN has exactly one input and one output place, the net is called \emph{State Machine}. Any 1-marked strongly connected State Machine is live and the maximum number of tokens in each place must be 1 for all reachable marking. A PN is an \emph{Choice Free} (CF) net if all places have at most one output transition. Finally, a PN is an \emph{Join Free} (JF) net if all transitions have at most one input place.

\begin{definition}\label{def:ssp}
Let $\mathcal{N}^s = \langle P,T,\b{Pre},\b{Post} \rangle$ be a PN net. The system $\langle \mathcal{N}^s,\b{m}_0\rangle$ is a \emph{Synchronized Sequential Processes} (SSP) if it can be decomposed into $n$ modules (called also agents) $\N_i^s = \langle P_i, T_i, \b{Pre}[P_i,T_i], \b{Post}[P_i,T_i] \rangle$ and a set of \emph{buffers} $B$ such that:
\begin{enumerate}
\item $P = P_1 \cup P_2 \cup \ldots \cup P_n \cup B$ and $P_i \cap P_j = \emptyset$ for all $i \neq j$ and $P_k \cap B = \emptyset$ for all $k=1,2, \ldots, n$.
\item $T = T_1 \cup T_2 \cup \ldots \cup T_n$ with $T_i \cap T_j = \emptyset$ for all $i \neq j$.
\item All $\N_i^s$ are strongly connected state machines.
\item All buffers in $B$ are destination private, i.e., for all buffers $b \in B$, if $\postset{b} \cap T_i \neq \emptyset$ then $\postset{b} \cap T_x = \emptyset$ for all $x = \{1,2, \ldots, n\} \setminus\{i\}$.
\item Each agent $i$ has only one marked place $p_i^e$, called \emph{waiting place}, such that all its cycles that have input buffers contain this place.
\item $\langle \N^s, \b{m}_0 \rangle$ is conservative and consistent.
\end{enumerate}
\end{definition}

Conditions 1) and 2) in Def. \ref{def:ssp} ensure that the set of places and transitions are partitioned into disjoint sets  (agents with their input buffers); condition 3) states that each agent is a strongly connected state machine  (thus locally consistent); condition 4) imposes that a buffer can only have output transitions in one agent (the destination agent). Conditions 5) considers only local cycles that have input buffers but remark that if a cycle without an input buffer has only output buffers the net cannot be conservative. Therefore, it is not possible to have cycles only with output buffers. Furthermore, local cycles without input and output buffers can be reduced to a transitions modeling local behaviours. Checking conditions in Def. \ref{def:ssp} can be done in  polynomial time. 

Notice that the class of systems of Def. \ref{def:ssp} is inspired from the DSSP \cite{SilvaDSSP}. On One side, Def. \ref{def:ssp} relaxes the assumption of DSSP that the buffer do not condition the choices in an agent. Another difference of Def. \ref{def:ssp} with DSSP is the constraint that represent the existence of a waiting place (condition 5)). Finally, condition 6) considers conservative and consistent SSP that are necessary conditions for structural liveness in structurally bounded PNs.

\begin{example}\label{example1}  The PN in Fig.~\ref{fig:DSSPmot} (without place $p_m$ and its input and output arcs) is a SSP (also a DSSP) modeling the production system represented schematically in Fig. \ref{fig:PracExam}. Agent 1, modeled by $\N_1$ is creating the engines ($t_1 \rightarrow p_2 \rightarrow t_2$ for car model A and $t_3 \rightarrow p_3 \rightarrow t_4$ for car model B). Agent 2, modeled by $\N_2$ is producing windshields ($t_5 \rightarrow p_5 \rightarrow t_6$ for car model A and $t_7 \rightarrow p_6 \rightarrow t_8$ for car model B). For simplicity, we assume that there exist raw material to produce only one type of intermediate product and a new one can be produced once it is consumed by agent 3. Agent 3 is producing the cars of type A (firing $t_9$) or type B (firing $t_{12}$). Notice that the type of the car is not chosen by resource availability but it is an outside decision (for example by a client). For this reason, in DSSP the buffers never constraint the internal choices. However, this restriction is removed in SSP since it is necessary to be violated for structural live enforcement. Agent 3 is producing first the bodywork ($t_9$ for type A and $t_{12}$ for type B), then is assembling the engine ($t_{10}$ or $t_{13}$ depending on the model) and finally the windshield ($t_{11}$ or $t_{14}$).

Each agent has two local T-semiflows. In particular, the T-semiflows of $\N_1$ are $\b{x}_{1}=t_1+t_2$ and $\b{x}_{2}=t_3+t_4$; the ones of $\N_2$  are $\b{x}_{3}=t_5+t_6$ and $\b{x}_{4}=t_7+t_8$; while the T-semiflows of $\N_3$ are $\b{x}_{5}=t_9+t_{10}+t_{11}$ and $\b{x}_{6}=t_{12}+t_{13}+t_{14}$. Notice that for sake of brevity, the multi-set notation for vectors is used. There exist four pairs of buffers $(b_{1}, b_ {5})$, $(b_{2}, b_{6})$, $(b_{3}, b_{7})$ and $(b_{4}, b_{8})$ in consumption - production relation. Finally, the waiting places are $p_1$, $p_4$ and $p_7$.

\ignore{The set of all equal conflict sets is: $SEQS = \{ \{t_1, t_3\},$ $\{t_2\},$ $\{t_4\},$ $\{t_5, t_7\},$ $\{ t_6\},$ $\{t_8\},$ $\{t_9, t_{12}\},$ $\{t_{10}\},$ $\{t_{11}\},$ $\{t_{13}\},$ $\{t_{14}\} \}$.
$\N^s$ does not fulfill the rank theorem: $ \mid SEQS\mid - 1 = 10 \neq rank (\b{C}) = 11$.}

This net is structurally not live and for the marking in Fig.~\ref{fig:DSSPmot}, by firing the sequence $t_3t_4t_5t_6t_3t_5t_9$, the SSP will reach the livelock (deadlock in this case) marking $\b{m}'=p_3+p_5+p_8+b_{1}+b_{4}+b_{6}+b_{7}$.
\end{example}

\section{On liveness of SSP nets}\label{sec:dssp}

\begin{figure*}[ht]
   \begin{center}
   \centering
	\psfrag{t1}{$t_1$}\psfrag{t2}{$t_2$}\psfrag{t3}{$t_3$}\psfrag{t4}{$t_4$}
\psfrag{t5}{$t_5$}\psfrag{t6}{$t_6$}\psfrag{t7}{$t_7$}\psfrag{t8}{$t_8$}
\psfrag{t9}{$t_9$}\psfrag{t10}{$t_{10}$}\psfrag{t11}{$t_{11}$}\psfrag{t12}{$t_{12}$}
\psfrag{t13}{$t_{13}$}\psfrag{t14}{$t_{14}$}
\psfrag{t15}{$t_{15}$}
\psfrag{p1}{$p_1$}\psfrag{p2}{$p_2$}
\psfrag{p3}{$p_3$}\psfrag{p4}{$p_4$}
\psfrag{p5}{$p_5$}\psfrag{p6}{$p_6$}
\psfrag{p7}{$p_7$}\psfrag{p8}{$p_8$}
\psfrag{p9}{$p_9$}\psfrag{p10}{$p_{10}$}
\psfrag{p11}{$p_{11}$}
\psfrag{b1}{$b_1$}\psfrag{b2}{$b_2$}
\psfrag{b3}{$b_3$}\psfrag{b4}{$b_4$}\psfrag{b5}{$b_5$}
\psfrag{N1}{$\N_1$}\psfrag{N2}{$\N_2$}
    \centering \subfigure[]{\includegraphics[width=0.6\columnwidth]{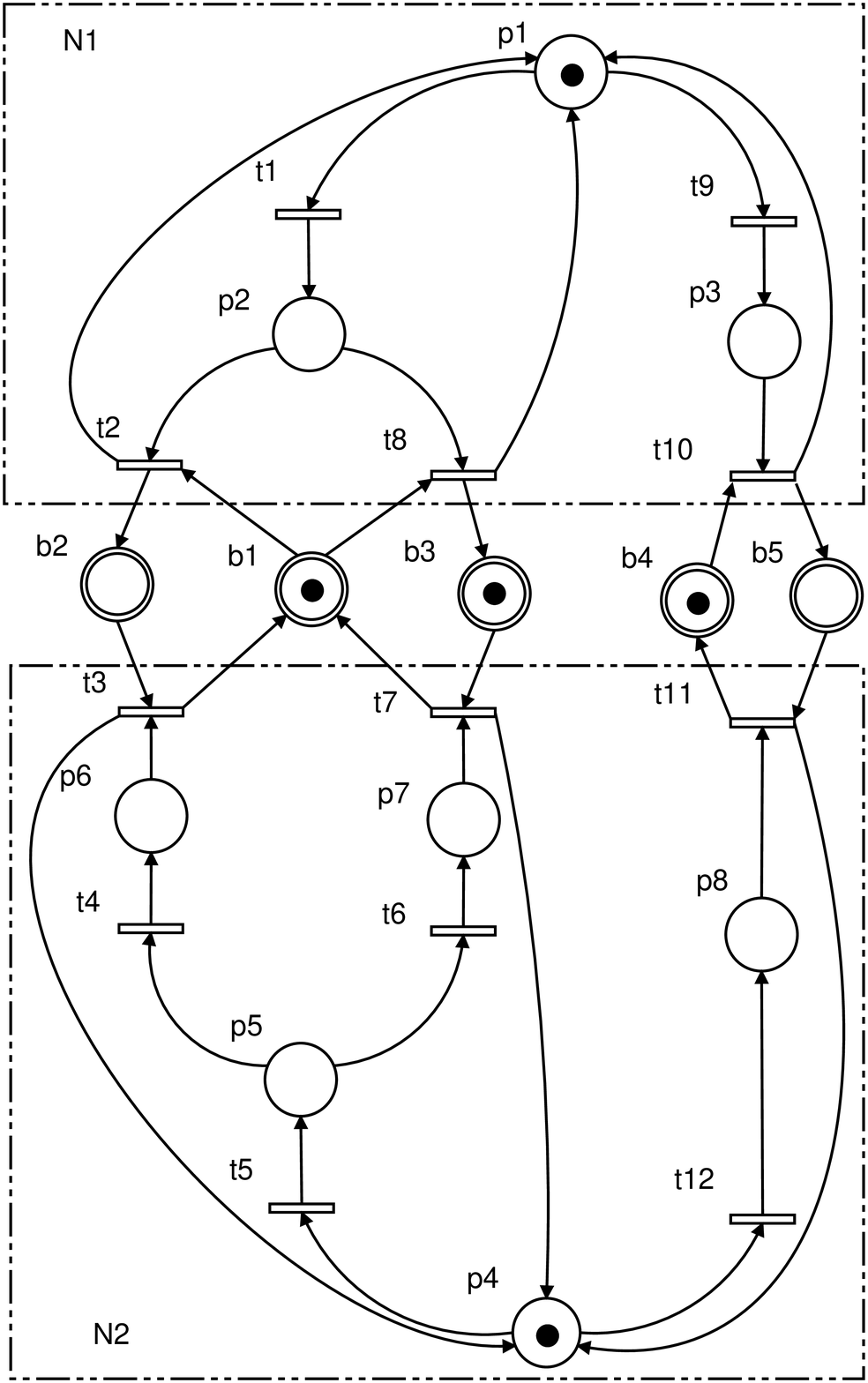}\label{fig:DSSPpb1}}\hspace{0.15\textwidth}
    \centering \subfigure[]{\includegraphics[width=0.6\columnwidth]{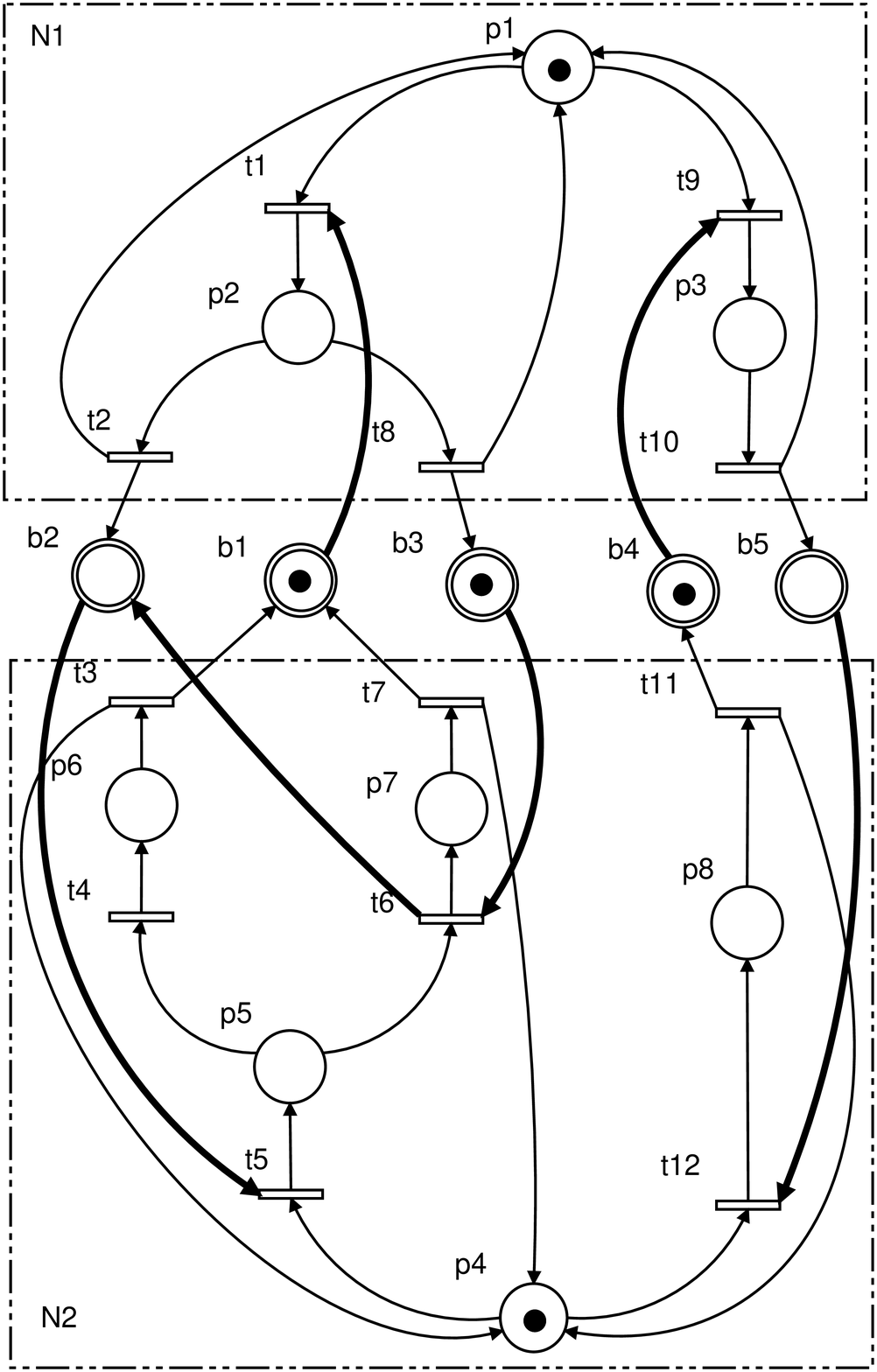}\label{fig:SSPPb1}}
\caption{Motivation example of problem 1: (a) A non live SSP net ($\N^s$); (b) The resulting non live net after applying the pre-assignment method in \cite{clavelCDC}. Observe that the number of buffers is not changing ($\{b_1, b_2, b_3, b_4, b_5\}$). Moreover, all the post-conditions remain unchanged, while transformations concern the pre-incidence, arc flowing from the buffers.} \label{fig:Pb1}
\end{center}
\end{figure*}

In this section, the main limitations of previous works (controlling bad siphons~\cite{ezpeleta,park2001deadlock,Colom:2003,li2004elementary,cano2012,7870672} and the method in \cite{clavelCDC}) for the liveness enforcement  of SSP net systems are considered. Moreover, some intuitions behind the approach proposed in this paper are provided.

\subsection{Controlling bad siphons in SSP systems} \label{siphonbased}
A well-known method for liveness enforcement of \emph{ordinary nets} consists in controlling the \emph{bad siphons}. First, the set of \emph{bad siphons} is computed and then they are prevented to become empty by using, for example, \emph{monitor} places (a kind of generalized mutual exclusion constraint~\cite{BOIoAn06,GMEC,ARLuWuZhSh20}). However, this strategy applied to SSP results in new buffers (the monitor places can be seen as new buffers) that may have output transitions in more than one agent. Therefore, the system loses the distributiveness property since the buffers are not anymore destination private (condition 4 of Def. \ref{def:ssp}). The place $p_m$ in Fig.\ref{fig:DSSPmot} prevents the emptiness of the bad siphon $S_1=\{p_1,p_2,p_7,p_9,p_{10},p_{11},b_2,b_5\}$. The new place $p_m$ can be seen as a new buffer, and it provides tokens to both agents $\N_1$ and $\N_3$. Moreover, the method of controlling bad siphons is well understood in ordinary Petri nets. However, SSPs are not necessarily ordinary. 

Additionally, deadlocks and circular waiting of resources is well understood for the special classes of RAS called S4PR \cite{IPTrGaCoEz05}. Nevertheless, the class of SSP is not comparable with the class of S4PR from RAS. The S4PR class is also composed by a set of state machine PNs connected through shared resources, hence something similar to the buffers. However, (i) In S4PR, resource places (corresponding to buffers in SSP) are not destination private; (ii) in S4PR there exists one P-semiflow for each resource place containing only that resource place while in SSP the P-semiflows may contain in general more than one buffer and (iii) in SSP the state machines contains only one token while in S4PR could exists more than one token in the idle place. Mainly because of (ii), the well know relation between circular waiting of resources and deadlock cannot be used in SSP. On the other hand, mainly because of (iii) the method presented here is difficult to be applied to S4PR or S3PR since in SSP only one token exists and in an agent, at the initial marking when some local T-semiflows can start firing only one can start in purely non-deterministic way.

Based on the previous reasons, new approaches for liveness enforcement in SSP should be taken into account. One such approach has been previously developed and it is briefly presented in subsection \ref{preinCDC}. However, it can be applied only to a reduced class of SSP. In subsection \ref{ss:controlpn} a new approach for more general SSP is introduced.

\subsection{Buffer pre-assignment in SSP}\label{preinCDC}
In~\cite{clavelCDC} an approach based on the pre-assignment of the buffers at the transitions in EQ relation\footnote{Two transitions $t$ and $t^\prime$ that are in conflict, i.e., $\preset t \cap \preset t^\prime \neq \emptyset$, are in EQ relation if $\b{Pre}[\cdot,t]= \b{Pre}[\cdot,t^\prime]$.} has been presented. The main idea of this method is to ensure that when a conflict transition is fired, at least one local T-semiflow that contains the corresponding conflict transition can be fired completely. 

\begin{example}
Fig.~\ref{fig:DSSPmotMon} shows the live net system obtained by applying the pre-assignment approach proposed in~\cite{clavelCDC} to the non-live SSP net $\N^s$ in Fig.~\ref{fig:DSSPmot} without place $p_m$. In $\N^s$, if transition $t_3$ is fired and $b_2$ is empty, agent $\N_1$ is blocked until $b_2$ receives a token. However, in Fig.~\ref{fig:DSSPmotMon} $t_3$ can be fired only if $b_2$ has a token because buffer $b_2$ has been preassigned from $t_4$ to $t_3$. Notice that in Fig. \ref{fig:DSSPmotMon} all the other buffers have been preassigned, i.e., $b_1$ is preassigned to $t_1$, $b_3$ to $t_5$, $b_4$ to $t_7$, \{$b_5$,$b_7$\} to $t_9$ and \{$b_6$,$b_8$\} to $t_{12}$.
\end{example}

Nevertheless, this method is not working for more general SSP structures due to the complex relations that may appear between local and global T-semiflows. In the following the two main problems (\textbf{Pb. 1} and \textbf{Pb. 2}) of the \emph{pre-assignment} method in \cite{clavelCDC} are stated and illustrated. Moreover, is given intuitions on how these problems will be approached in the new method presented in this paper. 

\textbf{Pb. 1.} If SSP net $\N^s$ has non disjoint global T-semiflows, it could happen that after applying the pre-assignment of buffers, the firing of a local T-semiflow is conditioned by the marking of a buffer that, in $\N^s$, was not its input buffer. This situation is changing the given plan and most frequently it makes the system too restrictive and possible non-live.  
  
\begin{example}\label{ex:Pb1}
Fig.~\ref{fig:DSSPpb1} shows a net $\N^s$ while Fig.~\ref {fig:SSPPb1} is the resulted net after applying the pre-assignment in~\cite{clavelCDC}. The global and local T-semiflows of $\N^s$ are given in Tab. \ref{table:TSFpb1} from where it is possible to check that $\b x_1$ and $\b x_2$ are not disjoint global T-semiflows ($\b x_1 \cap \b x_2=\{t_1,t_5\}$). 

\begin{table} [htbp]
\caption{Local and Global T-semiflows of $\N^s$ in Fig.~\ref{fig:DSSPpb1}}
\label{table:TSFpb1}
\begin{center}
\begin{tabular}{|c|c|c|c|c|}
\hline
Id. & Net & Type & Transitions & Global \\
\hline \hline
$\b x_1$ & $\N^s$ & Global & $t_1$-$t_5$ & - \\ \hline
$\b x_2$ & $\N^s$ & Global & $t_1$,$t_5$-$t_{8}$ & - \\ \hline
$\b x_3$ & $\N^s$ & Global & $t_9$-$t_{12}$ & - \\ \hline
$\b x_4$ & $\N_1^s$ & Local & $t_1$,$t_{2}$ & $\b x_1$  \\ \hline
$\b x_5$ & $\N_1^s$ & Local & $t_{1}$,$t_{8}$ & $\b x_2$  \\ \hline
$\b x_6$ & $\N_1^s$ & Local & $t_9$,$t_{10}$ & $\b x_3$  \\ \hline
$\b x_7$ & $\N_2^s$ & Local & $t_3$-$t_5$ & $\b x_1$  \\ \hline
$\b x_8$ & $\N_2^s$ & Local & $t_{5}$-$t_{7}$ & $\b x_2$  \\ \hline
$\b x_9$ & $\N_2^s$ & Local & $t_{11}$,$t_{12}$ & $\b x_3$  \\ \hline
\end{tabular}
\end{center}
\end{table}

In the resulted net in Fig. \ref{fig:SSPPb1}, $b_2$ has been preassigned to $t_5$, so the local T-semiflow $\b x_7=t_5+t_4+t_3$ can only start its firing when $b_2$ has a token. However, this pre-assignment is also conditioning the firing of the local T-semiflow $\b x_8=t_5+t_6+t_7$ since $t_5$ belongs to both $\b x_7$ and $\b x_8$. So, in the net in Fig. \ref{fig:SSPPb1}, the firing of the local T-semiflow $\b x_8$ is conditioned by $b_2$ that in $\N^s$ (Fig. \ref{fig:DSSPpb1}) was not its input buffer. This leads to block the execution of the global T-semiflows $\b x_1$ and $\b x_2$. Particularly, from the initial marking $\b m_0=p_1+p_4+b_1+b_3+b_4$ and by firing the sequence $t_1t_8$, $\b m_1=p_1+p_4+2\cdot b_3+b_4$ is reached. From this marking, $\b x_1$ and $\b x_2$ cannot fire anymore. 
\end{example}

In the approach proposed in this paper, in order to avoid \textbf{Pb 1}, a control PN in which the local T-semiflows are made disjoint is obtained. In particular, for each local T-semiflow in $\N^s$, a sequence $t_a \rightarrow p \rightarrow t_b$ is added in the control PN. In this way, the global T-semiflow are made disjoint and after the pre-assignment of the buffers in the control PN, the firing of a local T-semiflow only is conditioned by the buffers that in $\N^s$ were its input buffers.

\textbf{Pb. 2.} In the SSP net a buffer has to choose between the firing of different local T-semiflows that subsequently require a synchronization. Only considering the pre-assignment method, the net could remain non live. 

\begin{example}\label{ex:Pb2}
Fig.~\ref{fig:DSSPpb2} shows a $\N^s$ where the pre-assignment method in~\cite{clavelCDC} does not work due to \textbf{Pb. 2}. This $\N^s$ has a global T-semiflow $\b x_1$ composed by 3 local T-semiflows: $\b x_2=t_1+t_2$, $\b x_3=t_3+t_4$ and $\b x_4=t_5+t_6$. The firing of $\b x_2$($\b x_3$) consumes a resource from $b_1$ and produces a resource in $b_2$($b_3$). The firing of $\b x_4$ consumes a resource from $b_2$ and one from $b_3$ and produces two resources in $b_1$. So, in the long term (depending on the marking of the buffers), $\b x_2$, $\b x_3$ and $\b x_4$ should be fired proportionally. 
\begin{figure}[ht]
   \begin{center}
   \centering
	\psfrag{t1}{$t_1$}\psfrag{t2}{$t_2$}\psfrag{t3}{$t_3$}\psfrag{t4}{$t_4$}
\psfrag{t5}{$t_5$}\psfrag{t6}{$t_6$}\psfrag{t7}{$t_7$}\psfrag{t8}{$t_8$}
\psfrag{t9}{$t_9$}\psfrag{t10}{$t_{10}$}\psfrag{t11}{$t_{11}$}\psfrag{t12}{$t_{12}$}
\psfrag{t13}{$t_{13}$}\psfrag{t14}{$t_{14}$}
\psfrag{t15}{$t_{15}$}
\psfrag{p1}{$p_1$}\psfrag{p2}{$p_2$}
\psfrag{p3}{$p_3$}\psfrag{p4}{$p_4$}
\psfrag{p5}{$p_5$}\psfrag{p6}{$p_6$}
\psfrag{p7}{$p_7$}\psfrag{p8}{$p_8$}
\psfrag{p9}{$p_9$}\psfrag{p10}{$p_{10}$}
\psfrag{p11}{$p_{11}$}
\psfrag{b1}{$b_1$}\psfrag{b2}{$b_2$}
\psfrag{b3}{$b_3$}\psfrag{b4}{$b_4$}
\psfrag{N1}{$\N_1$}\psfrag{N2}{$\N_2$}
    \centering \subfigure[]{\includegraphics[width=0.45\columnwidth]{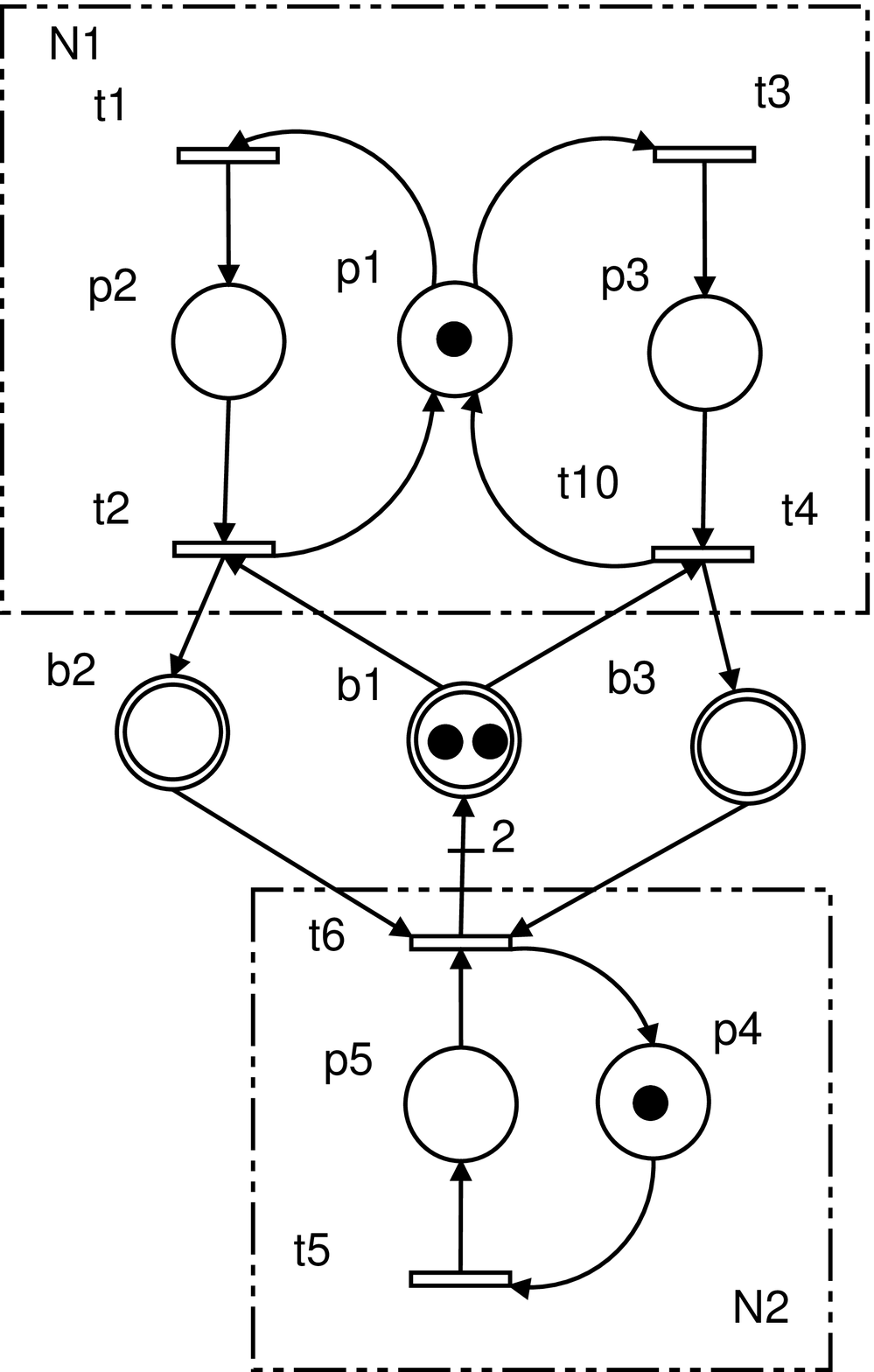}\label{fig:DSSPpb2}}
    \centering \subfigure[]{\includegraphics[width=0.45\columnwidth]{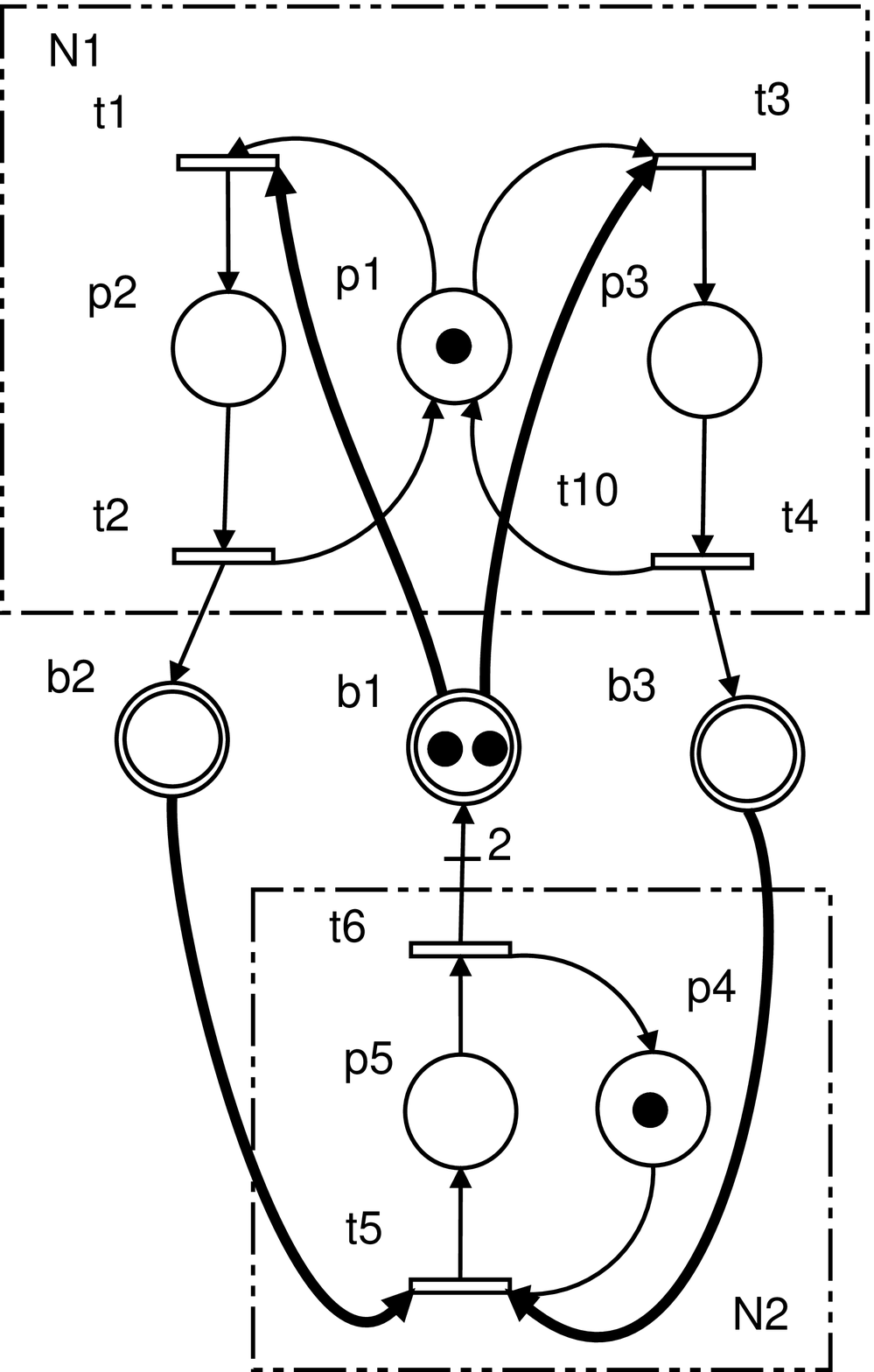}\label{fig:SSPpb2}}
\caption{Motivation example of problem 2: (a) A non live $\N^s$. (b) The resulting non live net after applying the pre-assignment method in \cite{clavelCDC}} \label{fig:Pb2}
\end{center}
\end{figure}

After applying the pre-assignment method to $\N^s$ in Fig.~\ref{fig:DSSPpb2}, the non-live net in Fig.~\ref{fig:SSPpb2} is obtained. Here it can be seen that $\b{x}_2$ or $\b{x}_3$ can be fired twice without a firing of $\b{x}_4$. Consequently, the resulting net is not live. 
\end{example}
In order to overcome \textbf{Pb. 2} in the new liveness enforcement approach, new buffers are included in the control PN, seen as \emph{information buffers}. The main objective of these new buffers is to force some global T-semiflows to fire their local ones in the correct proportion. These new buffers will only have output transitions in one local T-semiflow keeping the distributed property of the system.

\subsection{Structural Liveness enforcement through a control PN}\label{ss:controlpn}
Let us approach a new method with the idea of \emph{buffer pre-assignment} as the starting point.
However, unlike in~\cite{clavelCDC}, the pre-assignment here is not performed in $\N^s$, but in an \emph{control PN} denoted as $\N^c$. This $\N^c$ is obtained from $\N^s$ and has a predefined type of structure in which the local T-semiflows of $\N^c$ are disjoint. The main advantage of performing the pre-assigment in $\N^c$ with disjoint local T-semiflows is the fact that their will only be conditioned by buffers that in $\N^s$ were its input buffers. In this way, \textbf{Pb 1} is prevented.
Moreover, new buffers are included (seen as information buffers) forcing some global T-semiflows to fire their local ones in the correct proportion. These new buffers ensure the liveness of $\N^c$ preventing \textbf{Pb. 2}. 
\begin{figure}[ht]
\psfrag{Nc}{$\N^c$}
\psfrag{Nd}{$\N^s$}
   \begin{center}
   \includegraphics[width=1\columnwidth]{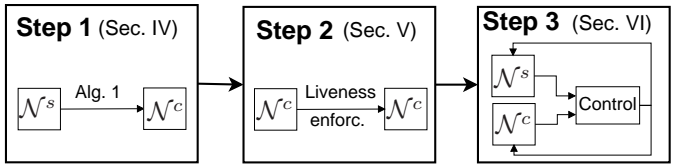}
\caption{\small Overview of the liveness enforcement methodology} \label{fig:StepsNdNc}
\end{center}
\end{figure} 

From a high level perspective (Fig. \ref{fig:StepsNdNc}), the proposed  methodology consists of:  

\begin{itemize}
\item \textbf{Step 1:} \emph{Compute the control PN $\N^c$} with a predefined structure and modeling the consumption/production relation between the buffers and the local T-semiflows. Sec.~\ref{sec:alg} explains how $\N^c$ is obtained by Alg.~\ref{alg:2}.

\item \textbf{Step 2:} \emph{Ensure the liveness of $\N^c$}. The control net $\N^c$ must be live. In Sec.\ref{sec:live} the possible structures of $\N^c$ obtained after applying Alg.~\ref{alg:2} are characterized. In some of them, the liveness holds while in others the liveness is forced by adding some control places (Alg.\ref{conpla}). % Unfortunately, there are some structures in which is not possible to fix a methodology to obtain a live control PN. 

\item \textbf{Step 3:} \emph{Control policy and systems evolution}. $\N^c$ will evolve synchronously with $\N^s$ disabling transitions that may lead the system to livelock. The methodology and behavior is described in Sec.~\ref{sec:contr}. \ignore{, basically consisting in, 

\begin{itemize}
\item \textbf{3.1} \emph{Label the transitions of the $\N^c$} with names of transitions of the $\N^s$. These common labels allow the firing of transitions in the $\N^c$ synchronously with some transitions of the $\N^s$.

\item \textbf{3.2} \emph{Assign guard expressions} to transitions of $\N^s$. These guard expressions are logical conditions based on the marking of $\N^c$ and will disable the firing of the transitions that may lead to a livelock. 

\item \textbf{3.3} Events from $\N^s$ are input signals to $\N^c$.
\end{itemize}}
\end{itemize}

\section{Construction of the control PN}\label{sec:alg}
\ignore{In Section \ref{preinCDC} we showed that buffer pre-assignment in a consistent and conservative SSP net $\N^s$ with non-disjoint local T-semiflows (\textbf{Pb. 1}), could result in non live systems.} This section presents an algorithm to compute the control PN denoted $\N^c$ for a given  structurally non live SSP $\mathcal{N}^s$. Subsequently and by means of guard expressions, the firing of local T-semiflows in $\N^s$ will be conditioned by the state (marking) in $\N^c$. Both net systems will evolve synchronously.

\ignore{
$\N^c$ is obtained from $\N^s$ and models the consumption/production relation that exists between the buffers and local T-semifows of $\N^s$. Moreover, $\N^c$ has the same number of agents and the same buffers of $\N^s$. However, each local T-semiflow in $\N^s$ is modeled by an ordinary sequence $t_a\rightarrow p\rightarrow t_b$ in $\N^c$. Where $t_a$ ($t_b$) models the first (last) transition of the local T-semiflow (defined in the following).

Before giving the methodology to obtain the control PN ($\N^c$) from a SSP structure ($\N^s$), let us state some concepts.
For computing $\N^c$ it is necessary that all agents composing $\N^s$ have a waiting place $p_e$ (Def.~\ref{def:wait}) from which the first and last transition of each local T-semiflow can be defined. The existence of a waiting place for each agent ensures that every cycle (local T-semiflow) contains a common input/output place ($p_e$) such that $p_e \in \postset{||\b x^i_n||} \cap \preset{||\b x^i_n||}$ for all T-semiflows $\b x^i_n$ of $\N_i$.}

\ignore{
\begin{figure}
\psfrag{t1}{$t_1$}\psfrag{t3}{$t_3$}\psfrag{t2}{$t_2$}\psfrag{t4}{$t_4$}
\psfrag{t6}{$t_6$}\psfrag{t5}{$t_5$}\psfrag{t7}{$t_7$}
\psfrag{p5}{$p_5$}\psfrag{p4}{$p_4$}\psfrag{p3}{$p_3$}\psfrag{p2}{$p_2$}
\psfrag{p1}{$p_1$}
   \begin{center}
   \includegraphics[width=0.8\columnwidth]{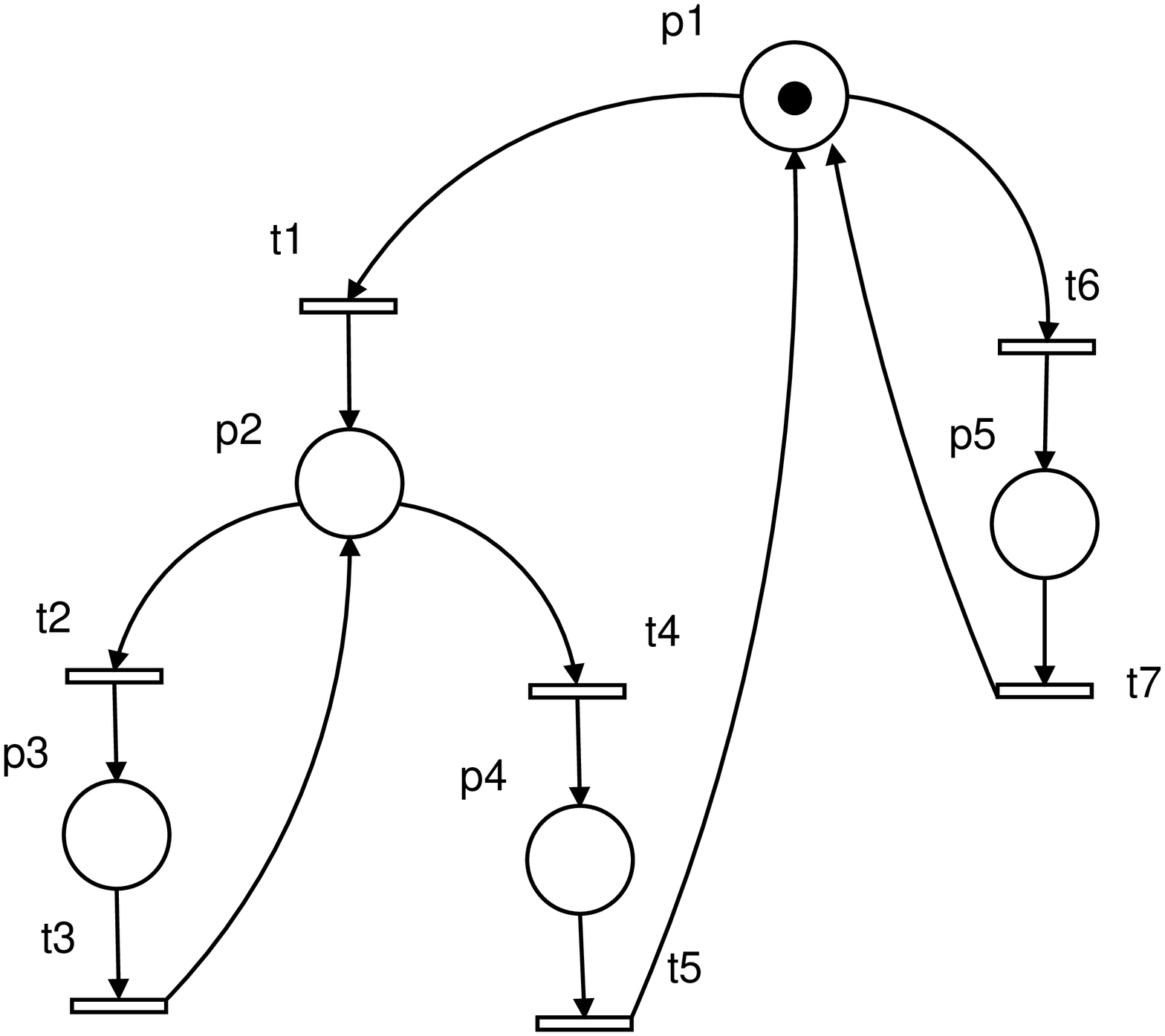}
\caption{\small Agent without waiting place} \label{fig:nowait}
\end{center}
\end{figure}  

\begin {definition} \label{def:wait}
Let $\N_i$ be an agent of a SSP $\N^s$, $\b x^i_1, \b x^i_2, \ldots, \b x^i_k$ be the minimal (local) T-semiflows of $\N_i$ and let $\bar{P_i}=\left(\displaystyle\bigcap_{n=1}^{k}\preset||\b x^i_n||\right) \cap P_i$ be the set of common places of the local T-semiflows, where $P_i$ is the set of places of $\N_i$ (i.e., without the buffers places)
\begin{itemize}
\item If $k=1$ (i.e., only one T-semiflow) the \emph{waiting place} is the marked place at the initial marking (which is unique according to Def.~\ref{def:dssp}).
\item If $k \geq 2$ and $\bar{P_i} \neq \b \emptyset$. The place $p_e \in \bar{P_i}$ is called \emph{waiting place} if there exists no direct path in $\N_i$ from $p_e$ to other place in $\bar{P_i} \setminus \{p_e\} $. 
\item If $k \geq 2$ and $\bar{P_i} = \emptyset$ implies that does not exists a \emph{waiting place} in $\N_i$.  
\end{itemize}  
\end{definition} 

For example, in the subnet $\N_1$ of the SSP net in Fig.~\ref{fig:DSSPpb1} there are three T-semiflows: $\b x_4$, $\b x_5$ and $\b x_6$ given in Tab.~\ref{table:TSFpb1}. $\preset||\b x_4|| \cap P_1=\{p_1,p_2\}$, $\preset||\b x_5|| \cap P_1=\{p_1,p_2\}$ and $\preset||\b x_8|| \cap P_1=\{p_1,p_3\}$. So $\bar{P_1}=p_1$ and consequently the waiting place is $p_1$.
However, the net in Fig.~\ref{fig:nowait} represents an agent of a SSP structure without a waiting place. This net have 3 minimal T-semiflow: $\b x_1=t_2+t_3$, $\b x_2=t_1+t_4+t_5$ and $\b x_3=t_6+t_7$. The input places of this minimal T-semiflow are the following: $\preset||\b x_1||=\{p_2,p_3\}$, $\preset||\b x_2||=\{p_1,p_2,p_4\}$ and $\preset||\b x_3||=\{p_1,p_5\}$. In this case, $\bar{P_i}=\emptyset$ and does not exist a waiting place.  } 
 
Let $p_i^e$ be the \emph{waiting place} of agent $\N_i$ that exists according to condition 5) of Def. \ref{def:ssp}. The first and last transitions of a local T-semiflow $\b x^i_l$ of agent $i$ are formally defined as,

\begin {itemize}
 \item $t^l_{first} = \postset {p_i^e} \cap ||\b x^i_l||$;
 \item $t^l_{last}= \preset {p_i^e} \cap ||\b x^i_l||$.
\end{itemize}

\begin{algorithm}[h]\label{alg:2}
 \begin{algorithmic}[1]
\REQUIRE SSP structure $\N^s=\langle P^s, T^s, \b{Pre}^s, \b{Post}^s \rangle$
\ENSURE Control PN $\N^c=\langle P^c, T^c, \b{Pre}^c, \b{Post}^c \rangle$
\STATE  Initialize the state of $\N^c$: $P^c := \emptyset, T^c := \emptyset , Pre^c :=\b 0 , Post^c := \b 0 ,$
\FORALL {$b_i$ $\in \N^s$}
\STATE Add a place $p_{b_i}$ to $P^c$
\ENDFOR
\FORALL {agents $\N_i^s$ of $\N^s$}
\STATE Add a place $p_{\N_i}$ to $P^c$
\STATE Compute all minimal T-semiflows of $\N_i^s$ in $\Gamma_i$ 
			\FORALL {$\b x^i_l \in \Gamma_i$}
			\STATE Add a transition $t^l_j$ to $T^c$; \COMMENT{representing $t^l_{first}$}
					\STATE Add a transition $t^l_k$ to $T^c$; \COMMENT{representing $t^l_{last}$}
					\STATE Add a place $p_{x_l}$ to $P^c$; \COMMENT{representing $\b x^i_l$} 
				  \STATE $\b {Post}^c[p_{x_l},t^l_j]=\b {Pre}^c[p_{x_l},t^l_k]=1$;
				  \STATE $\b {Pre}^c[p_{\N_i},t^l_j]= \b {Post}^c[p_{\N_i},t^l_k]=1$;
					
					\STATE Let $T_l = ||\b x^i_l||$;
					\FORALL {$b_i$ s.t. ${\b {Pre}^s}[b_i,T_l] \neq \b 0$}
					\STATE ${\b {Pre}^c}[p_{b_i},t^l_j]$= $\displaystyle\sum_{t \in T_l}{\b {Pre}^s}[b_i,t]$; 
					\ENDFOR
					\FORALL {$b_i$ s.t ${\b {Post}^s}[b_i,T_l] \neq \b 0$}
					\STATE ${\b {Post}^c}[p_{b_i},t^l_k]$= $\displaystyle\sum_{t \in T_i}{\b {Post}^s}[b_i,t].$ 
					\ENDFOR
\ENDFOR
\ENDFOR		

 \end{algorithmic}
 \caption{Computation of the control PN}
 \end{algorithm}
Using Alg.~\ref{alg:2}, a control PN $\N^c$ is obtained from a  structurally non live SSP structure $\N^s$ as follow: 

\begin{itemize}
\item First loop (steps 2-4): for each buffer $b_i$ in $\N^s$, a place $p_{b_i}$ is introduced in $\N^c$;
\item Second loop (steps 5-22): adds the other places and transitions to $\N^c$. This loop is iterated for all agents $\N_i^s$ of $\N^s$. Each iteration of $\N_i^s$ consists of,
\begin{itemize}
\item A new place $p_{\N_i}$ is added in $\N^c$ (step 6);
\item All minimal T-semiflows of $\N_i^s$ are computed  and saved in $\Gamma_i$ (step 7);
\item Third loop (steps 8 - 21) adds a ordinary sequence in $\N^c$ corresponding to each local T-semiflow of $\N_i^s$ and connect it with the input and output buffers. This loop consists of,
\begin{itemize}
\item For each local T-semiflow $\b x_l$ an ordinary subnet \{$t^l_j\rightarrow p_{x_l}\rightarrow t^l_k$\} is added in the $\N^c$ (steps 9-12);
\item Connect the place $p_{\N_i}$ with first transition of the ordinary sequence (introduced in step 9) and connect the last transition of the ordinary sequence (introduced in step 10) with place $p_{\N_i}$.
\item Step 14: All transitions belonging to the support of the T-semiflow $\b x_l$ are saved in $T_l$;
\item Forth loop (steps 15-17): for all input buffers of $\b x_l$ (input buffers of transitions in $T_l$), connect in $\N^c$ the corresponding input buffers (added in step 3) with the first transition of the ordinary sequence (added in step 9);
\item Fifth loop (steps 18-20): for all output buffers of $\b x_l$ (output buffers of transitions in $T_l$), connect in $\N^c$ the last transition of the ordinary sequence (added in step 10) with the corresponding output buffers (added in step 3);
\end{itemize}
\end{itemize}
\end{itemize}

In order to reduce the number of local T-semiflows in $\N^s$ and consequently the computational complexity of Alg.\ref{alg:2}, some basic and classical reduction rules~\cite{ICSilv93b,ARMura89} can be applied to $\N^s$ before computing $\N^c$. First, fusion of places and fusion of transitions must be considered and then, identical transitions (with the same input and output places) can be reduced to a unique transition.
\ignore{\begin{definition} Two transition $t$ and $t'$ are identical if $\b {Pre}[P,t]=\b {Pre}[P,t']$ and $\b {Post}[P,t]=\b {Post}[P,t']$. 
\end{definition}}

\begin{figure}
\psfrag{t1}{$\textcolor[rgb]{1,0,0}{t_1}$}\psfrag{t11}{$\textcolor[rgb]{1,0,0}{t_{11}}$}
\psfrag{t2}{$\textcolor[rgb]{1,0,0}{t_2}$}\psfrag{t10}{$\textcolor[rgb]{1,0,0}{t_{10}}$}
\psfrag{t5}{$\textcolor[rgb]{1,0,0}{t_5}$}\psfrag{t7}{$\textcolor[rgb]{1,0,0}{t_{7}}$}\psfrag{t12}{$\textcolor[rgb]{1,0,0}{t_{12}}$}
\psfrag{t6}{$\textcolor[rgb]{1,0,0}{t_6}$}\psfrag{t16}{$\textcolor[rgb]{1,0,0}{t_{16}}$}
\psfrag{t3}{$\textcolor[rgb]{1,0,0}{t_3}$}\psfrag{t13}{$\textcolor[rgb]{1,0,0}{t_{13}}$}
\psfrag{t9}{$\textcolor[rgb]{1,0,0}{t_9}$}\psfrag{t19}{$\textcolor[rgb]{1,0,0}{t_{19}}$}
\psfrag{t4}{$\textcolor[rgb]{1,0,0}{t_4}$}\psfrag{t14}{$\textcolor[rgb]{1,0,0}{t_{14}}$}
\psfrag{t8}{$\textcolor[rgb]{1,0,0}{t_8}$}\psfrag{t18}{$\textcolor[rgb]{1,0,0}{t_{18}}$}
\psfrag{x5}{$\textcolor[rgb]{1,0,0}{x_5}$}\psfrag{x6}{$\textcolor[rgb]{1,0,0}{x_6}$}
\psfrag{x7}{$\textcolor[rgb]{1,0,0}{x_7}$}\psfrag{x8}{$\textcolor[rgb]{1,0,0}{x_8}$}
\psfrag{x9}{$\textcolor[rgb]{1,0,0}{x_9}$}\psfrag{x10}{$\textcolor[rgb]{1,0,0}{x_{10}}$}\psfrag{x4}{$\textcolor[rgb]{1,0,0}{x_{4}}$}
\psfrag{x11}{$\textcolor[rgb]{1,0,0}{x_{11}}$}\psfrag{x12}{$\textcolor[rgb]{1,0,0}{x_{12}}$}
\psfrag{t41}{$t^4_1$}\psfrag{t42}{$t^4_2$}\psfrag{t51}{$t^5_1$}\psfrag{t56}{$t^5_8$}
\psfrag{t69}{$t^6_9$}\psfrag{t610}{$t^{6}_{10}$}\psfrag{t75}{$t^{7}_5$}
\psfrag{t73}{$t^7_{3}$}\psfrag{t85}{$t^8_{5}$}\psfrag{t87}{$t^{8}_{7}$}\psfrag{t912}{$t^{9}_{12}$}
\psfrag{t911}{$t^9_{11}$}
\psfrag{px4}{$p_{x_4}$}\psfrag{px5}{$p_{x_5}$}\psfrag{px6}{$p_{x_6}$}\psfrag{px7}{$p_{x_7}$}\psfrag{px8}{$p_{x_8}$}
\psfrag{px9}{$p_{x_9}$}
\psfrag{b1}{$p_{b_1}$}\psfrag{b2}{$p_{b_2}$}
\psfrag{b3}{$p_{b_3}$}\psfrag{b4}{$p_{b_4}$}
\psfrag{b5}{$p_{b_5}$}\psfrag{b6}{$p_{b_6}$}
\psfrag{b7}{$p_{b_7}$}\psfrag{b8}{$p_{b_8}$}\psfrag{N1}{$\N_1$}
\psfrag{N2}{$\N_2$}
\psfrag{bN1}{$p_{\N_1}$}\psfrag{bN2}{$p_{\N_2}$}
   \begin{center}
   \includegraphics[width=1\columnwidth]{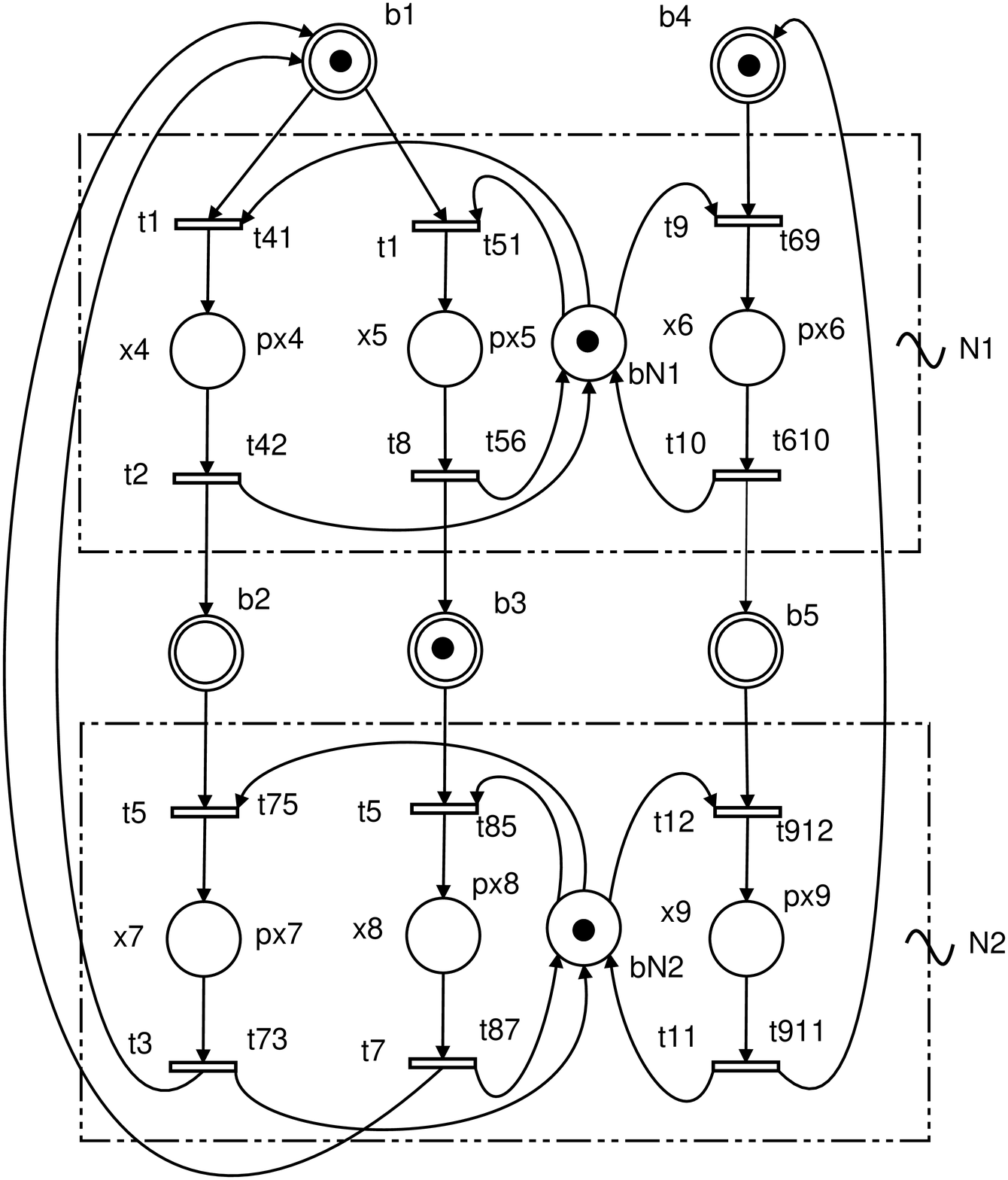}
\caption{\small Control PN obtained from SSP structure in Fig.~\ref{fig:DSSPpb1}} \label{fig:control1}.
\end{center}
\end{figure}

\begin{example}
Let us consider the SSP net $\N^s$ in Fig.~\ref{fig:DSSPpb1}. Applying Alg. \ref{alg:2} to $\N^s$, the control PN $\N^c$ in Fig.~\ref{fig:control1} is obtained. First, since $\N^s$ has 5 buffers, in $\N^c$ the places $p_{b_i}$ ($i=1,2,\dots,5$) are added. Furthermore, since $\N^s$ has two agents ($\N_1$ and $\N_2$), the places $p_{\N_1}$ and $p_{\N_2}$ are added in $\N^c$. After this, steps 8-20 are applied to all minimal T-semiflows of each agent. Let us consider for example the local T-semiflow $\b x_7= t_5+t_4+t_3$. Corresponding to this T-semiflow, the ordinary sequence \{$t^7_5\rightarrow p_{x_7}\rightarrow t^7_3$\} is added in $\N^c$. Since in $\N^s$, $b_2$ is an input buffer of $\b x_7$, in $\N^c$ there is an arc from the place $p_{b_2}$ to the input transition $t^7_5$. In addition, $\b x_7$ has an output buffer $b_1$ so in $\N^c$ there exits an arc from the output transition $t^7_3$ to the place $p_{b_1}$. Moreover, since $||\b x_7||$ belongs to $\N_2$, we add an arc from $p_{\N_2}$ to $t^7_5$ and from $t^7_3$ to $p_{\N_2}$. Finally, by applying steps 8-20 to all local T-semiflows, we obtain the control PN in Fig~\ref{fig:control1}. Notice that in Fig~\ref{fig:control1}, all six T-semiflows of $\N_1$ and $\N_2$ (given in Tab. \ref{table:TSFpb1}) are disjointly represented.
\end{example}

\ignore{\begin{figure}
\psfrag{t75}{$t^7_5$}\psfrag{t73}{$t^7_3$}
\psfrag{t3}{$t_3$}\psfrag{t4}{$t_4$}\psfrag{t5}{$t_5$}
\psfrag{te5}{$\textcolor[rgb]{1,0,0}{t_5}$}
\psfrag{te3}{$\textcolor[rgb]{1,0,0}{t_3}$}
\psfrag{px7}{$p_{x_7}$}\psfrag{p4}{$p_4$}\psfrag{p5}{$p_5$}\psfrag{p6}{$p_6$}
\psfrag{pb1}{$p_{b_1}$}\psfrag{pb2}{$p_{b_2}$}
\psfrag{b1}{$b_1$}\psfrag{b2}{$b_2$}
\psfrag{xe7}{$\textcolor[rgb]{1,0,0}{x_7}$}
\psfrag{pN2}{$p_{\N_2}$}\psfrag{N2}{$\N_2$}
\psfrag{Algorithm 2}{Algorithm 1}\psfrag{x5=t1,t2,t3}{$\b x_5=t_1,t_2,t_3$}

   \begin{center}
   \includegraphics[width=1\columnwidth]{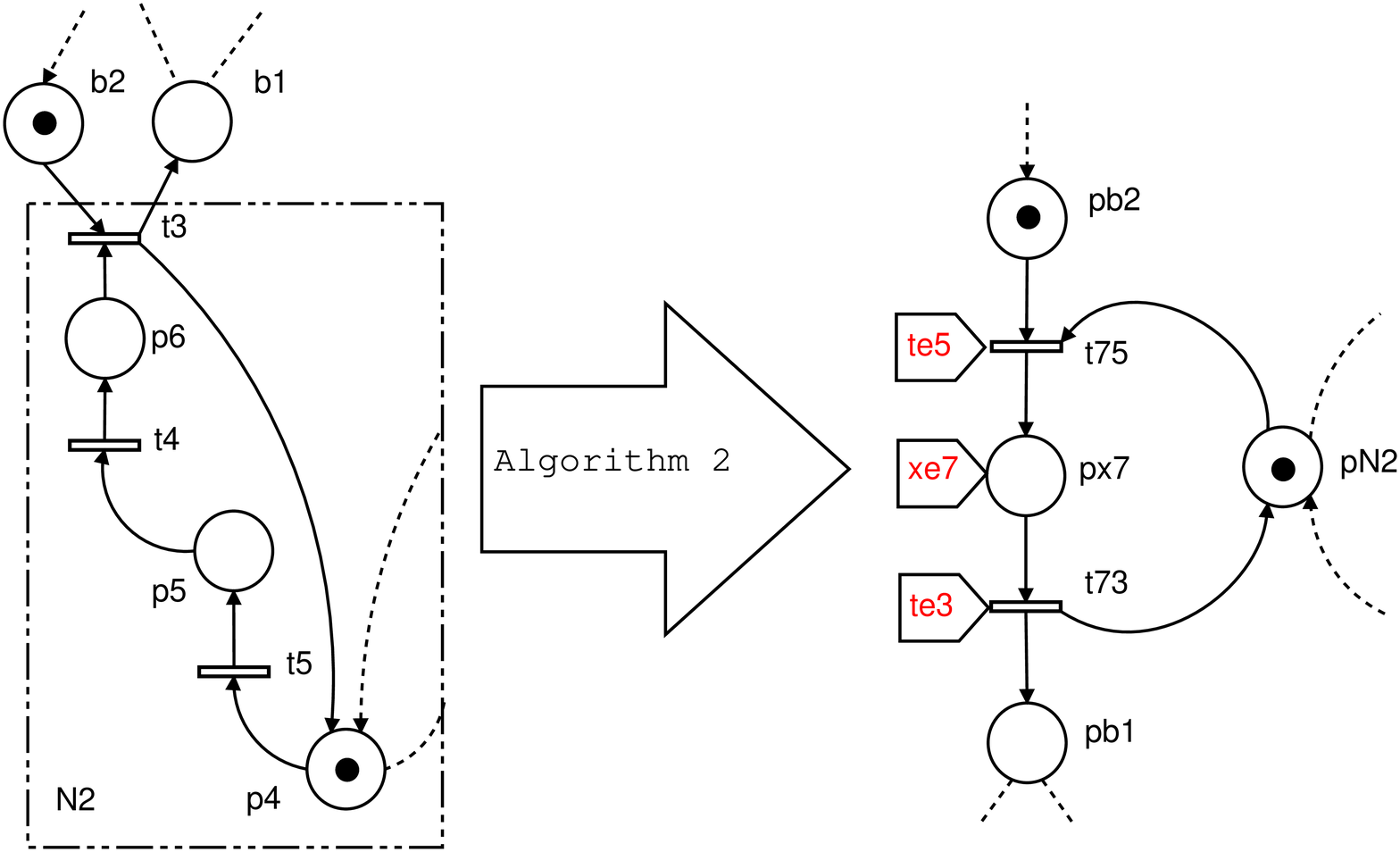}
\caption{\small Transformation obtained by applying algorithm \ref{alg:2} to a simple agent composed by 1 local T-semiflow $\b x_7=t_5+t_4+t_3$} \label{fig:transform}
\end{center}
\end{figure} 
}

\ignore{
\begin{figure*}[htb]
\psfrag{x}{$\b x$}\psfrag{tch}{$t_{ch}$}
   \begin{center}
   \includegraphics[width=1.7\columnwidth]{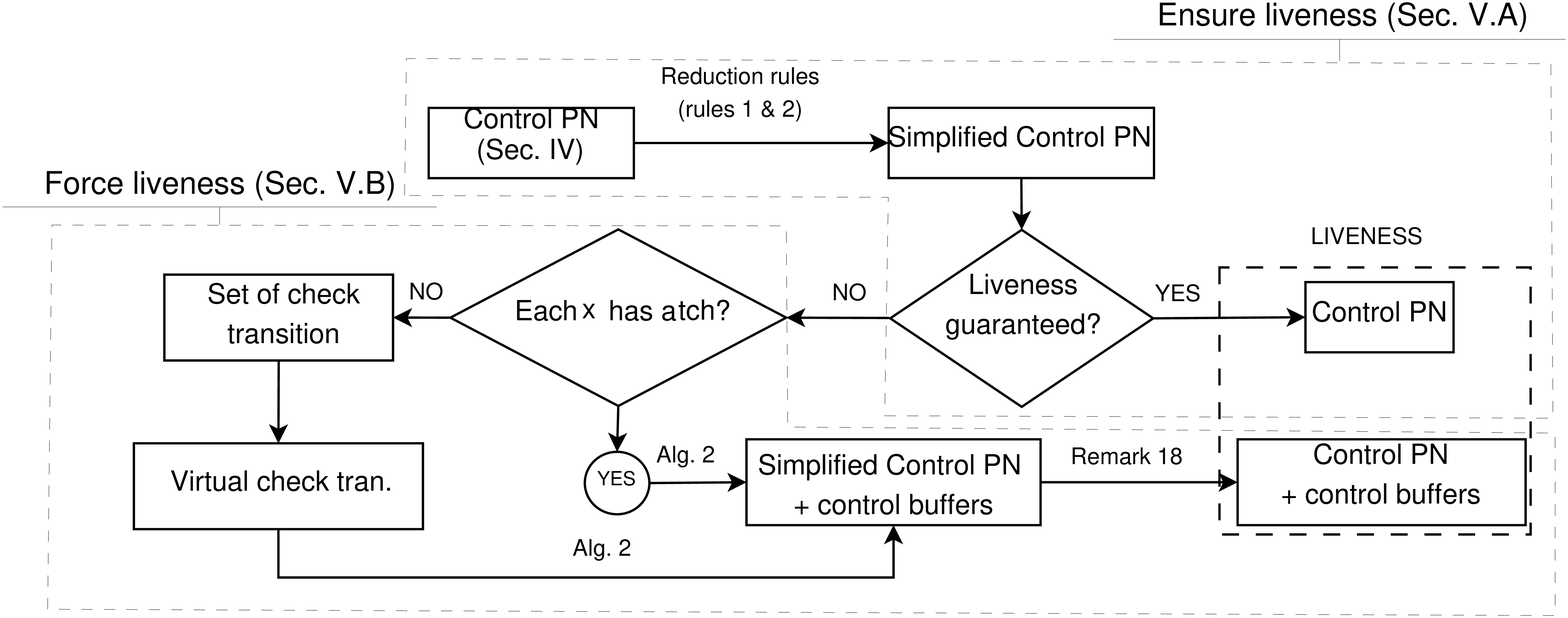}
   \caption{\small Diagram for ensuring liveness in the Control PN} \label{liveCPN}
\end{center}
\end{figure*}}

\section {Liveness of the Control PN}\label{sec:live}
Since the control PN system $\langle \N^c, \b{m}_0 \rangle$ should \emph{guide} the evolution of the SSP avoiding any blocking situation, it should be live as well. Otherwise, some local T-semiflows of the SSP will never be fired after control PN reaches a livelock state.

\subsection{Structural liveness analysis of the Control PN}\label{sec:liveen}
 To check structural liveness of the control PN, the following two reduction rules are first applied to the control PN $\N^c$ to obtain the \emph {simplified control PN} denoted as $\N^{cs}$.

\begin{rul}\label{reduc1}
The sequences \{$t^l_j\rightarrow p_{x_l}\rightarrow t^l_k$\} are reduced to a unique transition $t_{x_l}$.
\end{rul}

\begin{rul}\label{reduc2}
Once rule \ref{reduc1} is applied to all sequences, places $p_{\N_i}$ are implicit because $\b {Pre}^c[p_{\N_i},T]=\b {Post}^c[p_{\N_i},T]=\b m[p_{\N_i}]$ and are removed.
\end{rul}

Because in $\N^c$ a sequence \{$t^l_j\rightarrow p_{x_l}\rightarrow t^l_k$\} represents a local T-semiflow of $\N^s$, in $\N^{cs}$ transition $t_{x_l}$ (obtained by rule \ref{reduc1}) also represents a local T-semiflow of $\N^s$.
Moreover, by applying rule \ref{reduc2}, places $p_{\N_i}$ are removed, so all places in $\N^{cs}$ represent buffers.
In this way, $\N^{cs}$ may be composed of isolated subnets where the places represent buffers and the transitions represent local T-semiflows of $\N^s$. Moreover, each isolated subnet considers the global T-semiflows of $\N^s$ which have common input and/or output buffers. 
\begin{figure}[h]
   \begin{center}
	\psfrag{tx4}{$t_{x_4}$}\psfrag{tx5}{$t_{x_5}$}\psfrag{tx6}{$t_{x_6}$}\psfrag{tx7}{$t_{x_7}$}\psfrag{tx8}{$t_{x_8}$}\psfrag{tx9}{$t_{x_9}$}\psfrag{b1}{$p_{b_{1}}$}\psfrag{b2}{$p_{b_{2}}$}\psfrag{b3}{$p_{b_{3}}$}\psfrag{b4}{$p_{b_{4}}$}\psfrag{b5}{$p_{b_{5}}$}
   \includegraphics[width=0.5\columnwidth]{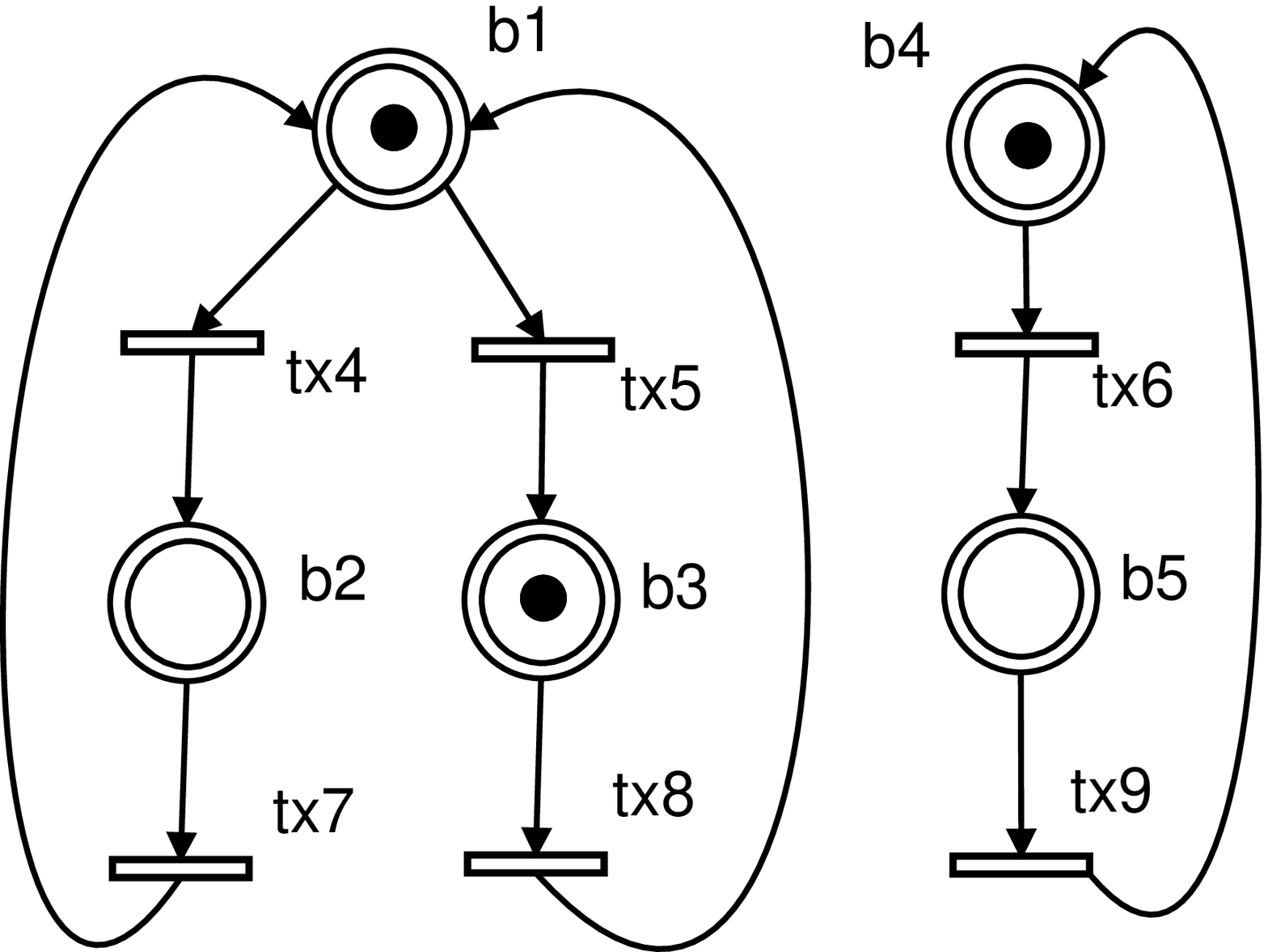}
   \caption{\small Simplified control PN $\N^{cs}$ after applying the reduction rules \ref{reduc1} and \ref{reduc2} to the control PN in Fig.\ref{fig:control1}.} \label{fig:ConSimPb1}
\end{center}
\end{figure} 

Fig. \ref{fig:ConSimPb1} shows the simplified control PN obtained after applying rules \ref{reduc1} and \ref{reduc2} to the control PN in Fig.\ref{fig:control1}. It can be seen that this net is composed by two isolated sub-nets. The one at left corresponding to the global T-semiflows $\b x_1$ and $\b x_2$ (given in Tab. \ref{table:TSFpb1}) which have the common input/output buffer $b_1$ (modeled by place $p_{b_1}$). The subnet at the right corresponds to the global T-semilfow $\b x_3$.

\begin{proposition} \label{def}
Let us assume a SSP structure $\N^s$, its control PN $\N^c$ obtained by applying Alg. \ref{alg:2}  and the simplified control PN $\N^{cs}$ obtained by applying Reduction Rules \ref{reduc1} and \ref{reduc2} to $\N^c$. If all isolated subnets $\N^{cs}_i$ of $\N^{cs}$ are CF or JF, then the control PN $\N^c$ is structurally live.
\end{proposition}
\begin{proof}
Alg.\ref{alg:2} preserves in $\N^c$ the number of agents, T-semiflows and buffers of $\N^s$. Moreover, Alg.\ref{alg:2} also preserves in $\N^c$ the consumption/production relation between buffers and local T-semiflows that exists in $\N^s$. Since $\N^s$ is consistent and conservative, then $\N^c$ is consistent and conservative. On the other hand, the reduction rules \ref{reduc1} and \ref{reduc2} do not change the number of tokens consumed and produced from and to the buffers, so if $\N^c$ is consistent and conservative, then each subnet  $\N^{cs}_i$ of $\N^{cs}$ is consistent and conservative.
Furthermore, if each subnet $\N^{cs}_i$ is CF or JF then $\N^{cs}$ is structurally live according to \cite{ARTeCoSi97}. 

On the other hand, the applied reduction rules (\ref{reduc1} and \ref{reduc2}) to $\N^{c}$ to obtain $\N^{cs}$ preserve liveness property \cite{ICSilv93b}. Therefore, if $\N^{cs}$ is structurally live then $\N^c$ is structurally live.
\end{proof}

For example, let us assume the simplified control PN of Fig. \ref{fig:ConSimPb1} obtained from the control PN in Fig. \ref{fig:control1}. Notice that the subnet on the left part is JF while the one on the right is CF and JF. According to Prop. \ref{def}, the corresponding control PN in Fig. \ref{fig:control1} is structurally live. therefore, there exists an initial marking for the buffers that make this net live. This marking could be for example the initial marking that is allowing the firing of all global T-semiflows once.

\subsection{Structural Liveness enforcement of the Control PN}\label{sec:livefor}

If a subnet $\N^{cs}_i$ is not JF neither CF then $\N^{cs}_i$ may not be structurally live. This subsection proposes a methodology to force the structural liveness of a structurally non live $\N^{cs}$. 

Recall that each $\N^{cs}_i$ of $\N^{cs}$ is consistent and conservative (see the first part of the proof of Prop.  \ref{def}). So, each $\N^{cs}_i$ is composed of one or more T-semiflows covering all transitions. In order to force structural liveness, the basic idea is to ensure that the transitions are fired proportionally according to the global T-semiflows. In this way, we prevent the free resolution of a conflict that subsequently requires a synchronization.

The proposed methodology adds an input place connected with ordinary arcs with each transition in conflict. The number of tokens in these new places is equal to the number of times that its output transition can be fired so, equal to the number of times that its output transition appears in each T-semiflows. 

In order to identify when a global T-semiflow has been fired completely, a check transition is introduced. Its firing means that the T-semiflow has finished and all its transitions can be fired again. So, the firing of the check transition should update the marking of the new added places.

Before giving formally the methodology to add this new control places let us consider a simple intuitive example.
\begin{example}
The left part of Fig.\ref{fig:ConSimPb2} shows the simplified control PN ($\N^{cs}$) of the SSP net $\N^s$ in Fig.\ref{fig:DSSPpb2}. This net is not live because there exists a conflict ($t_{x2}$,$t_{x3}$) and subsequently a synchronization in transition $t_{x4}$ (\textbf{Pb 2}).

\begin{figure}[h]
   \begin{center}
	\psfrag{tx4}{$t_{x_4}$}\psfrag{tx2}{$t_{x_2}$}\psfrag{tx3}{$t_{x_3}$}\psfrag{tx7}{$t_{x_7}$}\psfrag{tx8}{$t_{x_8}$}\psfrag{tx9}{$t_{x_9}$}\psfrag{b1}{$p_{b_{1}}$}\psfrag{b2}{$p_{b_{2}}$}\psfrag{b3}{$p_{b_{3}}$}\psfrag{b4}{$p_{b_{4}}$}\psfrag{b5}{$p_{b_{5}}$}\psfrag{Algorithm}{Alg. 2}\psfrag{ptx2}{$p_{t_{x2}}$}\psfrag{ptx3}{$p_{t_{x3}}$}
   \includegraphics[width=1\columnwidth]{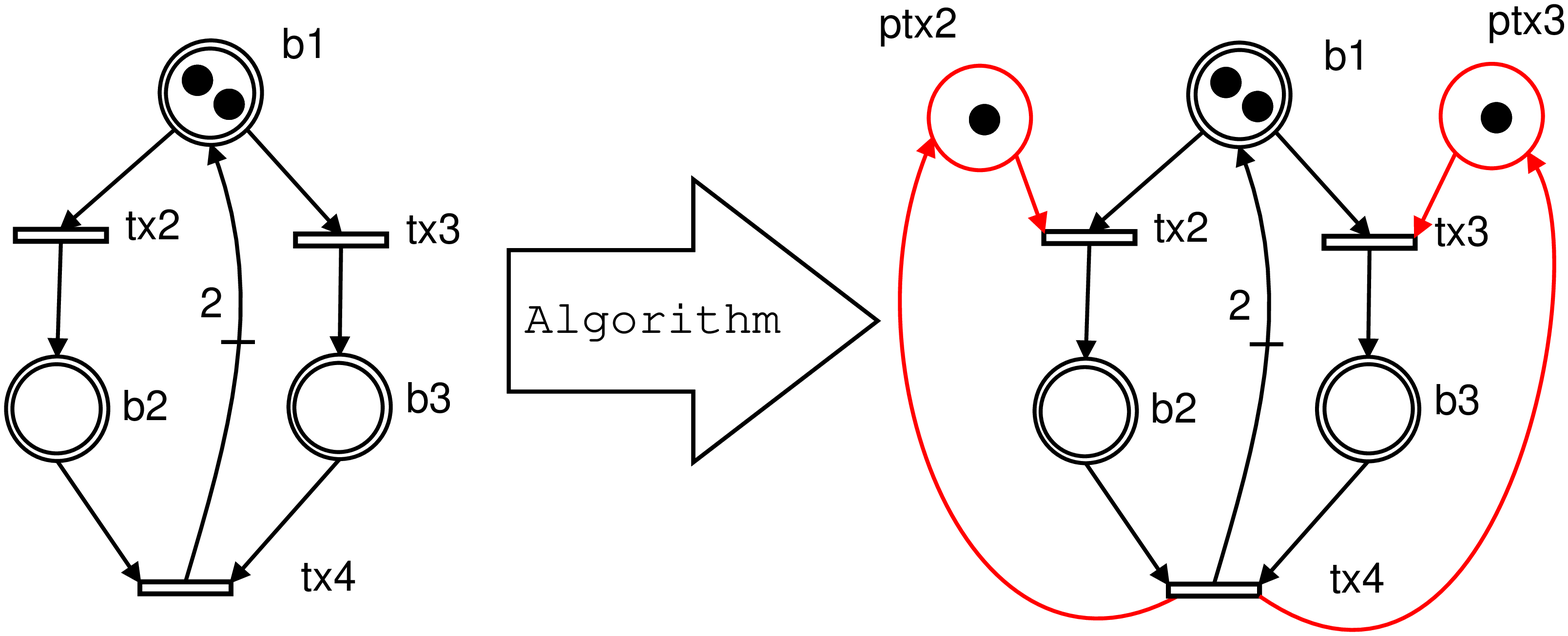}
   \caption{\small Left: simplified control PN of the SSP net in Fig. \ref{fig:DSSPpb2}; right: simplified control PN with new added buffers} \label{fig:ConSimPb2}
\end{center}
\end{figure} 

In the right part of Fig.\ref{fig:ConSimPb2} the new control buffers $p_{t_{x2}}$ and $p_{t_{x3}}$ are added for controlling the firing of $t_{x2}$ and $t_{x3}$. These new buffers constraint to one the times that $t_{x2}$ and $t_{x3}$ can fire before one firing of $t_{x4}$.
\end{example}

\ignore{To know when a global T-semiflow of a subnet composing the simplified Control PN has been fired completely, we define its check transition} 
The firing of a check transition of a given global T-semiflow must produce enough tokens to fire completely the T-semiflow again. For this net, the check transition is $t_{x4}$. The formal definition of check transitions is given in Def.\ref{cecktr}.

\begin{definition}\label{cecktr}
 Let $\N^{cs}_i$ be a subnet of a simplified control PN $\N^{cs}$ and let $\b x_j$ be a local T-semiflow of $\N^{cs}_i$. A transition $t^j_{ch} \in ||\b{x}_j||$ is called \emph{check transition} of $\b x_j$ if 
 \begin{itemize}
\item  it is not belonging to the support of other T-semiflow of $\N^{cs}$ and,
\item its firing creates enough tokens to fire again completely $\b x_j$. 
\end{itemize}
If there are more than one transition  in $||\b{x}_j||$ that fulfill the previous two constraints, any non conflict transition can be chosen. 
\end{definition}
\ignore{being $\b m= \b {Post}^{cs}[P,t^j_{ch}]$ the marking produced by firing of $t^j_{ch}$, the following conditions are satisfied:
\begin{enumerate}
\item $\begin{array}{l}
\b x_k[t^j_{ch}]= 
\left\{
\begin{array}{lll}
1,  &\text{if} &k=j \\
0,  &\text{if} &k \ne j
\end{array}
\right.,
\end{array}$
\item There exists a firing sequence $\sigma$ such that its firing vector $\b{\sigma}$ satisfies $\b m \xrightarrow{\sigma} \b m$ and $\b{\sigma}=\b x_j$.
\end{enumerate} }

\ignore {\red Condition 1) in the previous definition ensures that the check transition $t^j_{ch}$ belongs only to $||\b{x}_j||$ while condition 2) ensures the existence of a firing sequence with the firing vector equal to $\b{x}_j$ that can be fired after the firing of $t^j_{ch}$. If there are more than one transition that fulfill the constraints in Def. \ref{cecktr}, any non conflict transition can be chosen. }

\begin{algorithm}[h]\label{conpla}
 \begin{algorithmic}[1]
\REQUIRE A structurally non live net $\N^{cs}_i$.
\ENSURE A structurally live subnet $\N^+$ and a live $\b m_0$.
\STATE Let $\N^+ = \N_i^{cs}$.
\STATE Compute the set $X$ of minimal T-semiflows of $\N^{cs}_i$.
\FORALL {T-semiflow $\b x_i$ of $X$}
\STATE Obtain the check transition $t^i_{ch}$.
\ENDFOR
\STATE Let $T_{ck}$ be the set of check transitions.
\STATE Let $T_{cn}$ be the set of conflict transitions.
\FORALL {$t_i$ $\in$ ($T_{cn}$ $\backslash$ $T_{ck}$)}
\STATE Add a place $p_{t_i}$ s.t. $\b{Pre}^+ [p_{t_i},t_i]=1$.
\ENDFOR	
\FORALL {T-semiflow $\b x_i$ of $X$}
	\FORALL {$t_j$ s.t $\b x_i[t_j]>0$ and $t_j \in (T_{cn}$ $\backslash$ $T_{ck}$)}
    \STATE $\b {Post}^+[p_j,t^i_{ch}]=\b x_i[t_j]$
	\ENDFOR
\ENDFOR
\STATE  $\b m_0=\displaystyle\sum_{\b x_i \in X}{\b {Post}^+[P^+,t^i_{ch}]}$
 \end{algorithmic}
 \caption{Enforcing liveness of a structurally non live $\N^{cs}_i$.}
 \end{algorithm}

Alg. \ref{conpla} adds the control places and compute an initial marking that ensure the liveness of a subnet $\N^{cs}_i$ by ensuring the proportional firing of the transitions belonging to the T-semiflows. Each new added place only has one output transition, so the resulted net preserves the distributiveness property. 
It is necessary that all T-semiflows of the sub-nets to which Alg. \ref{conpla} is applied to have a check transition according to Def. \ref{cecktr}. Alg. \ref{conpla} adds a control place $p_{t_j}$ for each conflict transition $t_j$ that is not a check transition $t^i_{ch}$. The firing of $t_j$ is constrained by the marking of the new added place $p_{t_j}$: $\b{Pre}^+ [p_{t_j},t_j]=1$ (steps 8-10).
When the check transition $t^i_{ch}$ of T-semiflow $\b x_i$ is fired, the marking of the places $p_{t_j}$ that control transition $t_j$ belonging to $\b x_i$ are updated: $\b {Post^+}[p_j,t^i_{ch}]=\b x_i[t_j]$ (Steps 11-15). The initial marking $\b m_0$ is given by the marking that would be produced by firing each check transition $t^i_{ch}$ (step 16).

\begin{figure}
\psfrag{pt1}{$\red{p_{t_1}}$}\psfrag{pt2}{$\red{p_{t_2}}$}\psfrag{pt3}{$\red{p_{t_3}}$}\psfrag{pt4}{$\red{p_{t_4}}$}
\psfrag{pt5}{$\red{p_{t_5}}$}\psfrag{pt6}{$\red{p_{t_6}}$}\psfrag{pt7}{$\red{p_{t_7}}$}
\psfrag{t1}{$t_1$}\psfrag{t2}{$t_2$}\psfrag{t3}{$t_3$}\psfrag{t4}{$t_4$}\psfrag{t5}{$t_5$}\psfrag{t6}{$t_6$}\psfrag{t7}{$t_7$}\psfrag{t8}{$t_8$}\psfrag{t9}{$t_9$}\psfrag{p1}{$p_1$}\psfrag{p2}{$p_2$}\psfrag{p3}{$p_3$}\psfrag{p4}{$p_4$}\psfrag{p5}{$p_5$}\psfrag{p6}{$p_6$}\psfrag{p7}{$p_7$}\psfrag{p8}{$p_8$}\psfrag{b1}{$b_{1}$}\psfrag{b2}{$b_{2}$}\psfrag{b3}{$b_{3}$}\psfrag{b4}{$b_{4}$}\psfrag{x1}{$\b{x}_{1}$}\psfrag{x2}{$\b{x}_{2}$}\psfrag{x3}{$\b{x}_{3}$}\psfrag{x4}{$\b{x}_{4}$}\psfrag{N1}{$\N_1$}\psfrag{N2}{$\N_2$}
   \begin{center}
   \includegraphics[width=1\columnwidth]{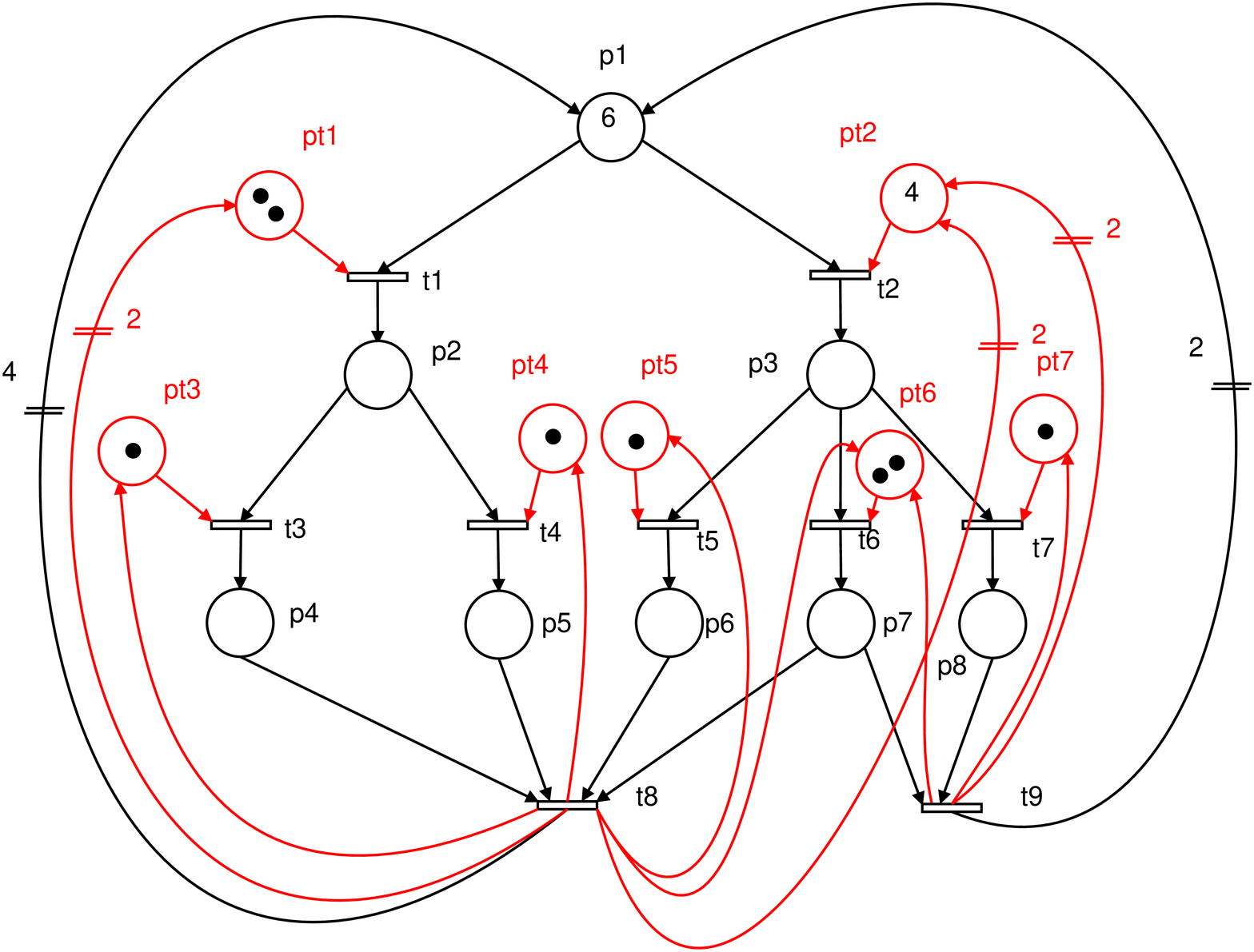}
\caption{\small A non live $\N^{cs}_i$ and the control places (red).} \label{Simcpn}
\end{center}
\end{figure}

In Fig. \ref{Simcpn}, a possible subnet $\N^{cs}_i$ (containing both choices and synchronizations) of a simplified control PN $\N^{cs}$is shown in black. This net is composed by two global T-semiflows $\b x_1=2t_1+2t_2+t_3+t_4+t_5+t_6+t_8$ and $\b x_2=2t_2+t_6+t_7+t_9$. Moreover, the control places $p_{t_j}$ added by applying Alg.\ref{conpla} are represented in red color. 
$\N^{cs}_i$ without the control places is not live because there are choices between \{$t_1, t_2$\}, \{$t_3, t_4$\} and \{$t_5,t_6,t_7$\} and subsequently synchronizations in $t_8$ and $t_9$ are required (\textbf{Pb.} 2). However, after adding the control places, some places become implicit. Particularly, adding $p_{t_1}$ and $p_{t_2}$ with $\b m_0[p_{t_1}]=2$ and $\b m_0[p_{t_2}]=4$ the transition $t_1$ can be initially fired two times and transition $t_2$ four times becoming the place $p_1$ implicit. Removing $p_1$ (implicit) it is possible to combine $p_{t_1}$ with $p_2$ and $p_{t_2}$ with $p_3$. Now, the control places $p_{t_3}$ - $p_{t_7}$ make the previous combined places implicit. Repeating this process iteratively, it is possible to check that the simplified control PN with the added places is live. 

\begin{lemma}
The net system resulted after applying Alg.\ref{conpla} to a structurally non live subnet $\N^{cs}_i$ is live.
\end{lemma}
\begin{proof}
$\N^{cs}_i$ is a consistent and conservative PN (see the proof of Prop. \ref{def}) where each transition represents a local T-semiflow of $\N^s$ and each place represent a buffer. The initial marking $\b{m}_0$ of the net is equal to the marking produced by firing of check transitions.

By applying Alg.\ref{conpla}, a control place $p_{t_j}$ is added for each conflict transition $t_j$ that is not a check transition. These control places limit the firing of conflict transitions making the original input places of these conflicts to be implicit. This happens because the firing of check transitions produce,
\begin {itemize}
\item in its original output places (buffers), enough tokens to fire completely the global T-semiflows.
\item in the new added places $p_{t_j}$, the exact number of tokens for firing the conflict transitions $t_j$ the times that appears in the global T-semiflows. 
\end {itemize}

If a place that receives tokens by firing of check transitions is not a decision (i.e., has only one output transition), it can be fused with the output places of its unique output transition. This procedure can be iterated until the obtained place is a decision. Moreover, the marking of this place is greater or equal to the marking of the added places $p_{t_j}$ (i.e., constraining the firing of the transitions in the conflict). In this way, this place is implicit and can be removed without compromise the liveness of the net.  After removing of an implicit place, its output transitions $t_j$ are not anymore in conflict and each one has only one input place $p_{t_j}$. Starting now with an input place $p_{t_j}$ the full procedure can be iterated and at the end it will be removed being implicit for the next conflict. Finally, the net structure obtained is composed by the check transitions connected with one place by a self-loop. So, the resulting net system is live and consequently the original net system with the added places is live.
\end{proof}

Alg.\ref{conpla} adds some control places and computes an initial marking that force the liveness of the simplified Control PN ($\N^{cs}$). However, the control PN ($\N^c$) is the net used on the control level. So, the new added control places in $\N^{cs}$  should be translated to $\N^c$.
\begin {remark}\label{refine}
Considering that the sequence \{$t^l_j\rightarrow p_{x_l}\rightarrow t^l_k$\} in $\N^c$ has been reduced to a unique transition $t_{x_l}$ in $\N^{cs}$, for each new control buffer ($b_s$) added in $\N^{cs}$ an homologous buffer ($b_h$) is added in $\N^c$ as follows: 
\begin{itemize}
\item $\b {Pre}^c[b_h,t^l_j]=\b {Pre}^{cs}[b_s,t_{x_l}]$
\item $\b {Post}^{c}[b_h,t^l_k]=\b {Post}^{cs}[b_s,t_{x_l}]$
\item $\b m_0[b_h]=\b m_0[b_s]$
\end{itemize}
\end{remark}

\textbf{Absence of check transition.}
Alg. \ref{conpla} assumes that for each T-semiflow of $\N^{cs}_i$ it is possible to compute a check transition according to Def. \ref{cecktr}. Unfortunately, this is not always true, and could exists no check transition for some T-semiflows in $\N^{cs}_i$.
In these situations should be necessary to fix a \emph{set of check transitions} that together fulfill the conditions in Def. \ref{cecktr}. One possible brute force approach for finding a set of check transitions could be to analyze all possible sub-sets of transitions bellonging to the support of the analyzed T-semiflow.
Once a set of check transitions is fixed, the \emph{new virtual check transition} can be added. This new transition must synchronize the set of check transitions. Fig. \ref{checktran} shows a possible set of check transitions (\{$t_1, t_2, t_3$\}) and the new virtual check transition $t_{ch}$. 

\begin{figure}[htb]
\psfrag{t1}{$t_1$}\psfrag{t2}{$t_2$}\psfrag{t3}{$t_3$}\psfrag{t4}{$t_{ch}$}\psfrag{set}{Set of check transitions}\psfrag{add}{Virtual check transition}
   \begin{center}
   \includegraphics[width=0.3\columnwidth]{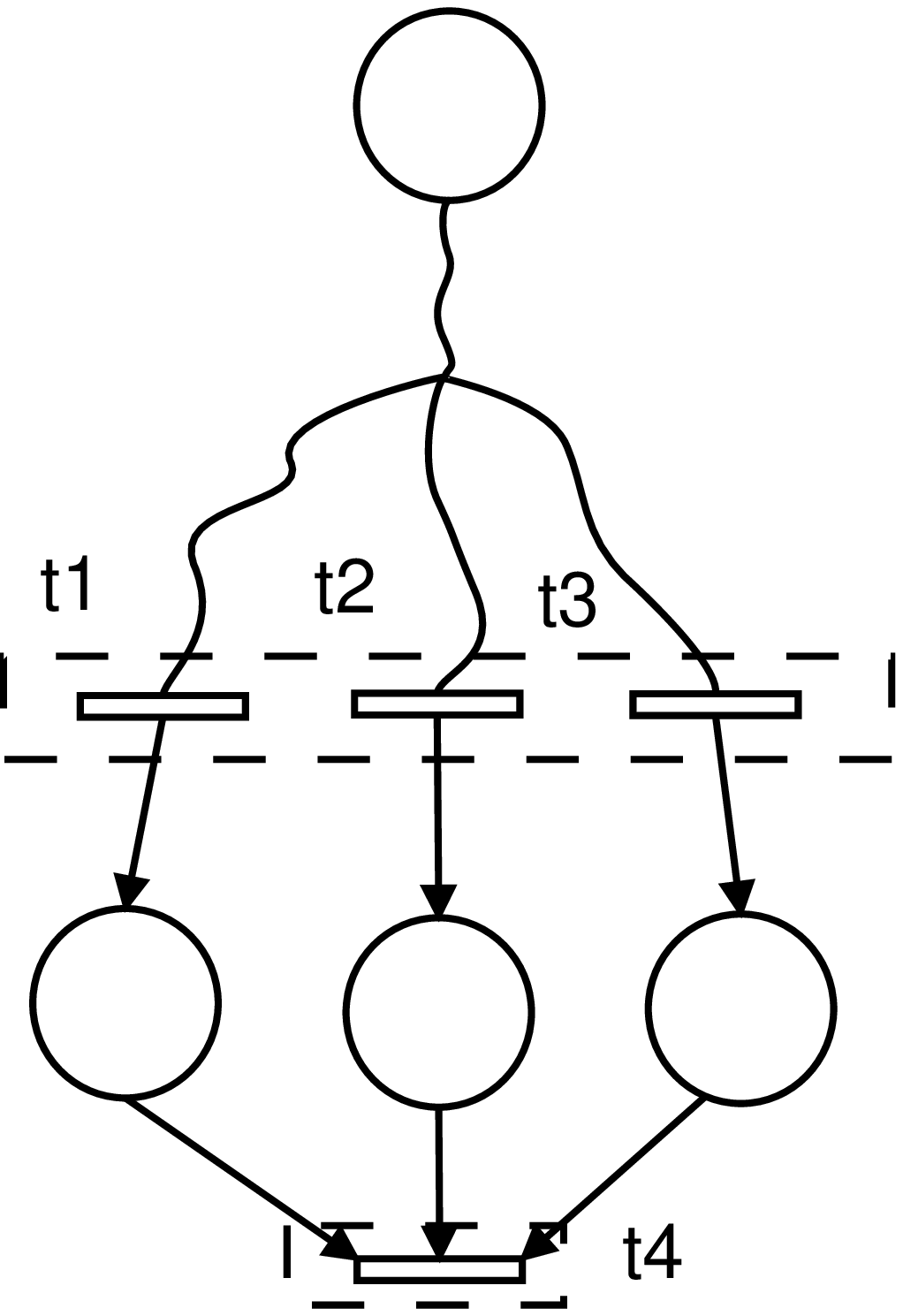}
   \caption{\small Adding a virtual check transition $t_{ch}$.} \label{checktran}
\end{center}
\end{figure}

\section{Control policy and system evolution}\label{sec:contr}
This section discusses the control policy based on a live control PN system. The idea is based on the following two conditions,
\begin{condi} \label{condi1}
 The firing of a local T-semiflow can only starts if all its input buffers have enough tokens to complete the firing of all its transitions.
\end{condi}
\begin{condi} \label{condi2}
When the firing of a local T-semiflow $\b x^i_l$ starts in the agent $\N_i$, it is not possible to fire transitions that are not in the support of $\b x^i_l$ until all transitions of $||\b x^i_l||$ in $\N_i$ are fired.
\end{condi}
To achieve that $\N^c$ constraint the fireability of the transitions in $\N^s$ guard expressions are used. Each transition in $\N^s$ has associated a guard expression. These guard expressions are logical conditions related with the marking of $\N^c$. A transition $t \in T^s$ can be fired only if the associated guard expression is true. Of course, to fire $t$, it is also necessary to be enabled in $\N^s$.

\emph{Labelling in Control PN}.
Each transition in $\N^c$ is labelled with the name of a transition of $\N^s$. In the process of obtaining $\N^c$, for each local T-semiflow $\b x_l$ in $\N^s$ an ordinary subnet ($t^l_j \rightarrow p_{x_l}\rightarrow t^l_k$) in $\N^c$ has been added. The transition $t^l_j$/$t^l_k$ is labelled with the name of the first/last transition of $\b x_l$. In addition, the place $p_{x_l}$ is labelled with the name of the local T-semiflow $\b x_l$.

\emph{Control policy}. A transition $t$ belonging to $\N^s$ can be fired if it is enabled ($\b m^s \geq \b {Pre}^s[P,t]$) and its guard expression is true.   
The guard expression associated with  $t$ becomes true if at least one of the next conditions is fulfilled in $\N^c$:
\begin{itemize}
\item there is a transition labelled as $t$ and it is enabled. 
\item the marking of a place labelled with the name of a local T-semiflow  to which $t$ belongs is equal to one. 
\end{itemize}
First condition prevents to start firing a local T-semiflow whose input buffers do not have enough tokens. While second condition ensures that if a local T-semiflow has been started, all the guard expressions of its transitions become true.
   
\emph{Evolution of the control PN}. The evolution of $\N^c$ is synchronized with the evolution of $\N^s$ (see Fig.~\ref{fig:ControlDiagram} for a schematic view). The events in $\N^s$ (firing of transitions) are considered as input signals in $\N^c$. When a transition $t_k$ is fired in $\N^s$, then in $\N^c$:
\begin{itemize}
\item if no transition labelled $t_k$ exists, no transition is fired.
\item if there exist enabled transitions labelled $t_k$, the controller (scheduler) chooses one and is fired in $\N^c$. 
\end{itemize}

\begin{algorithm}[h]\label{policy}
 \begin{algorithmic}[1]
\REQUIRE A SSP $\langle \N^{d}= \langle P^{d}, T^{d}, \b{Pre}^{d}, \b{Post}^{d} \rangle, \b m^s_0 \rangle$ and its live control PN system $\langle \N^{c}= \langle P^{c}, T^{c}, \b{Pre}^{c},$ $\b{Post}^{c} \rangle, \b m^c_0 \rangle$.
\ENSURE A live evolution of $\langle \N^{d}, \b{m}_0^s \rangle$ through $\langle \N^{c}, \b{m}_0^c \rangle$
\STATE $\b m^s:= \b m^s_0$ \COMMENT{initializing the current state in $\N^{d}$}
\STATE $\b m^c:= \b m^c_0$ \COMMENT{initializing the current state in $\N^{c}$}
\STATE $T^e:= \emptyset$ \COMMENT{initializing the set of enabled transitions in $\N^s$ at marking $\b m^s$}
\STATE $T^f:= \emptyset$ \COMMENT{initializing the set of transition that can be fired in $\N^s$}
\FORALL {$t_i \in T^s$}
\IF{$\b m^s \geq \b {Pre}^s[P^s,t_i]$}
\STATE $T^e:=T^e \cup t_i$;
\ENDIF
\ENDFOR
\FORALL {$t_j \in T^e$}
\IF{the guard of $t_j$ is true at marking $\b m^c$}
\STATE $T^f:=T^f \cup t_j$
\ENDIF
\ENDFOR
\STATE A transition $t_k \in T^f$ is fired in $\N^s$
\STATE {If there exists $T^s \subseteq T^f$ such that all $t_o \in T^s$ have the same label in $\N^c$ as $t_k$, select one  as $t_k$.}
\STATE $\b m^s:= \b m^s+\b C^s[P^s,t_k]$ \COMMENT{$\b m^s$ is updated}
\IF{there exists a transition labelled as $t_k$ in $\N^c$}
\STATE a transition labelled as $t_k$ is fired in $\N^c$ 
\STATE $\b m^c:= \b m^c+\b C^c[P^c,t_k]$ \COMMENT{$\b m^c$ is updated}
\ENDIF
\STATE go to Step 3
 \end{algorithmic}
 \caption{Control policy and systems evolution on a SSP net $\N^s$ and its control PN $\N^c$.}
 \end{algorithm}

Alg. \ref{policy} implements the control policy and the system evolution on a SSP net and its control PN.

\begin{figure}[ht]
\psfrag{Input signal (t)}{Input signal ($t_k$)}\psfrag{Event (t)}{Event ($t_k$)}
\psfrag{enabeled}{enabled}
\psfrag{Control}{Control}
\psfrag{m}{$\b m^s$}\psfrag{mc}{$\b {m^c}$}
\psfrag{Te}{$T^e$}\psfrag{Tf}{$T^f$}
\psfrag{Nc}{$\N^c$}\psfrag{Nd}{$\N^s$}
\psfrag{DSSP}{SSP}
\psfrag{Enabeled}{Enabled}\psfrag{transitions}{transitions}
   \begin{center}
   \includegraphics[width=1\columnwidth]{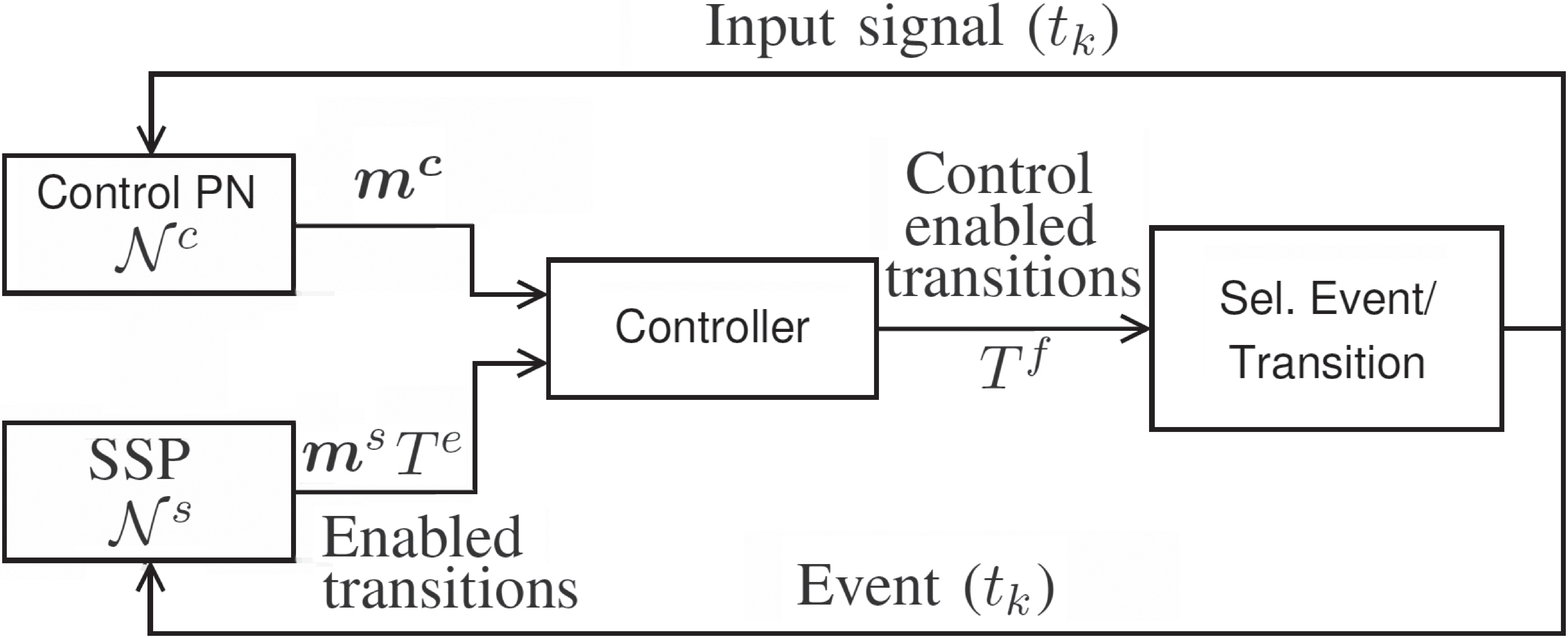}
\caption{\small Control diagram of a SSP structure using a control PN} \label{fig:ControlDiagram}
\end{center}
\end{figure} 

\ignore{Fig.~\ref{fig:ControlDiagram} shows the control flow diagram of the proposed liveness enforcement approach for a SSP structure ($\N^s$) using a control PN ($\N^c$). The control policy depends on the markings of $\N^s$ and $\N^c$. The marking of $\N^c$ inhibits some enabled transitions to be fired in $\N^s$ if a deadlock could appear. A transition $t$ is chosen from the set of control-enabled transitions.}

\ignore{
\begin{figure*}
 \begin{center}
\centering
\psfrag{pt1}{$\red{p_{t_1}}$}\psfrag{pt2}{$\red{p_{t_2}}$}\psfrag{pt3}{$\red{p_{t_3}}$}\psfrag{pt4}{$\red{p_{t_4}}$}
\psfrag{pt5}{$\red{p_{t_5}}$}\psfrag{pt6}{$\red{p_{t_6}}$}\psfrag{pt7}{$\red{p_{t_7}}$}
\psfrag{t1}{$t_1$}\psfrag{t2}{$t_2$}\psfrag{t3}{$t_3$}\psfrag{t4}{$t_4$}\psfrag{t5}{$t_5$}\psfrag{t6}{$t_6$}\psfrag{t7}{$t_7$}\psfrag{t8}{$t_8$}\psfrag{t9}{$t_9$}\psfrag{p1}{$p_1$}\psfrag{p2}{$p_2$}\psfrag{p3}{$p_3$}\psfrag{p4}{$p_4$}\psfrag{p5}{$p_5$}\psfrag{p6}{$p_6$}\psfrag{p7}{$p_7$}\psfrag{p8}{$p_8$}\psfrag{b1}{$b_{1}$}\psfrag{b2}{$b_{2}$}\psfrag{b3}{$b_{3}$}\psfrag{b4}{$b_{4}$}\psfrag{x1}{$p_{x_{1}}$}\psfrag{x2}{$p_{x_2}$}\psfrag{x3}{$p_{x_3}$}\psfrag{x4}{$p_{x_4}$}\psfrag{N1}{$\N_1$}\psfrag{N2}{$\N_2$}\psfrag{bN1}{$p_{\N_1}$}\psfrag{bN2}{$p_{\N_2}$}\psfrag{bN3}{$p_{\N_3}$}\psfrag{N3}{$\N_3$}
 
    \centering \subfigure[]{\includegraphics[width=0.75\columnwidth]{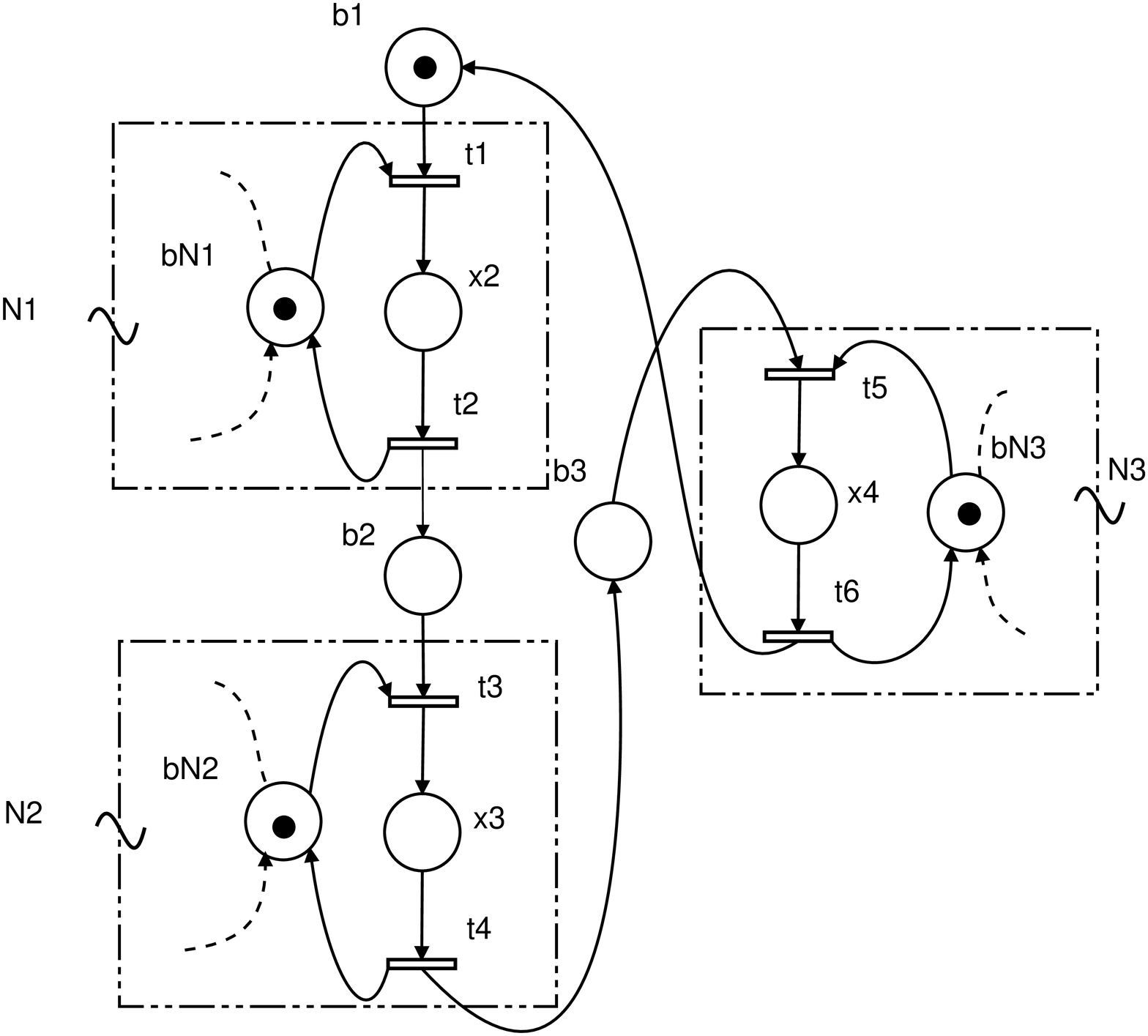}\label{fig:AdvanceOut}}\hspace{0.1\textwidth}
    \centering \subfigure[]{\includegraphics[width=0.75\columnwidth]{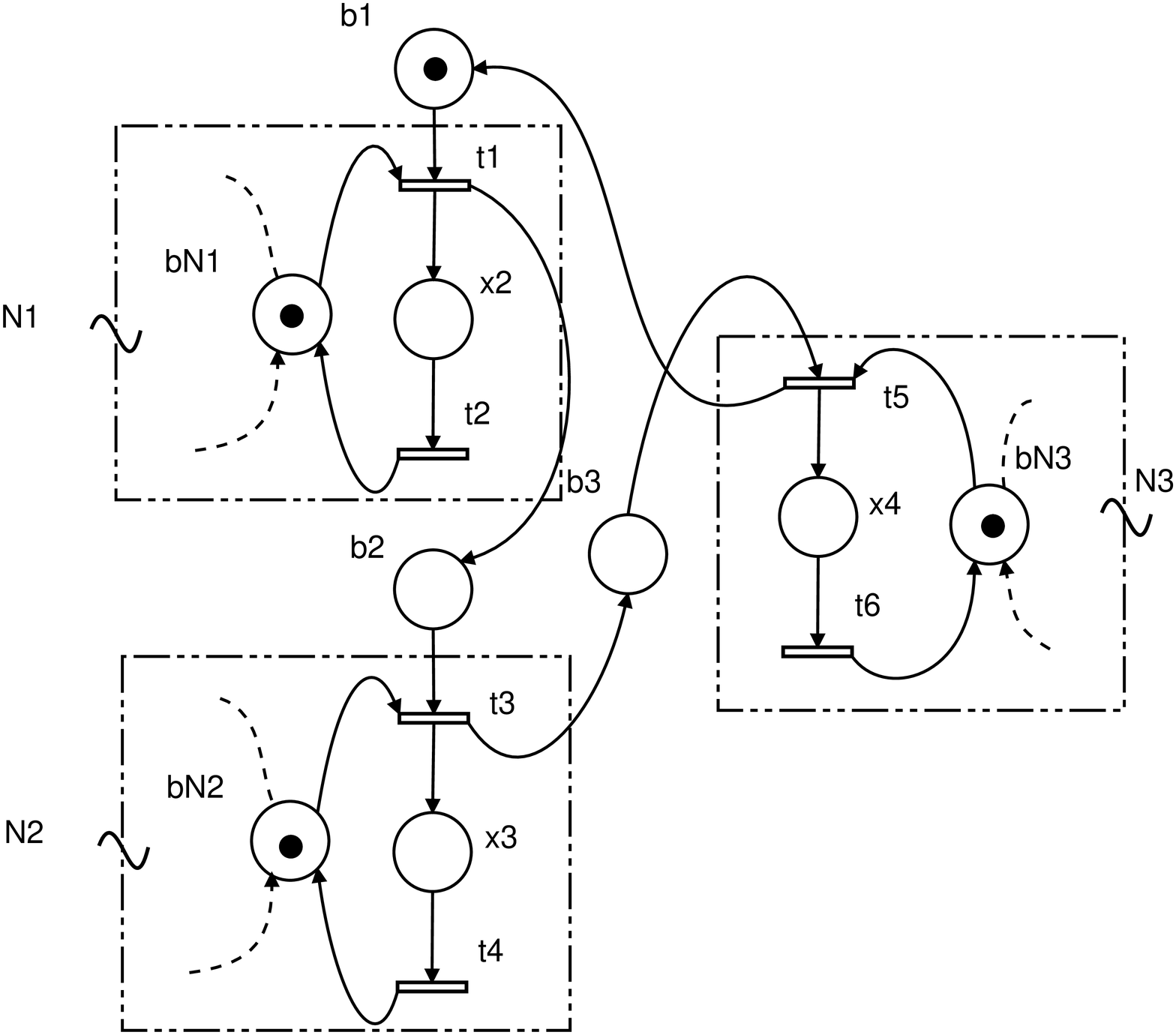}\label{fig:AdvanceOut2}}
\caption{\small A Control PN where: a) the production of tokens is generated from the last transition of the sequences; b) the production of tokens is generated from the first transition of the sequences} \label{fig:Advance}
\end{center}
\end{figure*}}

\begin{theorem}
Let $\N^s$ be a non live SSP structure and $\N^c$ the structurally live control PN obtained by Alg. \ref{alg:2}, eventually after applying Alg. \ref{conpla} to force its structural liveness. If the initial marking of buffers in $\N^s$ is greater than or equal to to the initial marking of buffers that makes $\N^c$ live, by following the control policy described in Alg.~\ref{policy}, the controlled SSP is live.
\end{theorem}
 
\begin{proof}
If the control PN $\N^c$ is structurally live, then by definition there exists an initial marking $\b{m}_0$ that makes the net system $\langle \N^c, \b{m}_0 \rangle$ live. Since $\N^c$ is representing the relations between the T-semiflows of $\N^s$ and buffers, putting the marking of the buffers in $\N^s$ equal to the marking of buffers that makes $\N^c$ live will allow the firing of all global T-semiflows in isolation (not implying that $\N^s$ is live).

The control policy in Alg. \ref{policy}, computes first the set of enabled transitions of SSP $\N^s$ in $T^e$ (steps 5-9). Then set $T^f$ is obtained from $T^e$ by removing those transitions with guard expression not valid (steps 10-14).  All transitions in $T^f$ are enabled in $\N^s$ and their guard expressions are true. Guard expressions are \emph{True} if Conditions \ref{condi1} and \ref{condi2} are true. 

Condition \ref{condi1} ensures that if a first transition of a local T-semiflow is in $T^f$ then its firing means that a local T-semiflow is starting to fire and there exists enough tokens in all input buffers to completely fire it. If the first transition of a local T-semiflow is fired in $\N^s$ (is chosen from the set $T^f$ in steps 15-17) then the transition labelled with the same name is also fired in $\N^c$ (steps 18-21), and consequently a token is generated in a place labelled with the name of the local T-semiflow in $\N^c$. Notice that, if exist more transitions labelled with the same labelled in $\N^c$, then the scheduler select one to fire (step 16). 

If a transition of $T^f$ is not a first transition of a local T-semiflow, Condition \ref{condi2} ensures that it belongs to the support of the local T-semiflows that haves already been started to fire. This is done by simple checking the marking of the corresponding place of $\N^c$. When a last transition of a local T-semiflow is fired, the waiting place is marked again and a new local T-semiflow could start firing.

As $\N^s$ is consistent and conservative, there exist sequences of firing of the local T-semiflow to fire global T-semiflows. The relations between the T-semiflows and buffers are modeled in the control PN $\N^c$. Since $\langle \N^c, \b{m}_0 \rangle$ is live it will allow to fire only the local T-semiflows that will not produce livelocks forcing the liveness of controlled SSP.
\end{proof}

The computational complexity of the proposed live enforcement approach is exponential because it is necessary compute the set of global minimal T-semiflows. However, this computation should be done only once, at the beginning of the approach and it is not necessary to be iterated, as for example in the case of controlling bad siphons.

\begin{example}

In $\N^s$ of the Fig.~\ref{fig:DSSPpb1}, transitions $t_1$, $t_{9}$, $t_5$ and $t_{12}$ are enabled, but only $t_1$, $t_{9}$ and $t_5$ can be fired because in its $\N^c$ given in Fig.~\ref{fig:control1} transition labelled as $t_{12}$ ($t^{9}_{12}$) is not enabled. In this way, only the local T-semiflows whose input buffers have enough tokens for completing their firing are allowed to start. Let us assume that $t_5$ (Fig.~\ref{fig:DSSPpb1}) is fired in $\N^s$:
\begin {itemize}
\item automatically the enabled transition $t^8_5$ (labelled as $t_5$) has to be fired in $\N^c$. 
\item Now in $\N^s$, the transitions $t_1$, $t_{9}$, $t_4$ and $t_{6}$ are enabled. 
\item In $\N^c$ transition labelled as $t_1$ ($t^4_1$ and $t^5_1$) and $t_9$ ($t^6_9$) are enabled, so $t_1$ and $t_9$ can be fired in $\N^s$. Moreover, in $\N^c$ the marking of the place $p_{x8}$ labelled as $x_8$ is equal to 1, so transition $t_{6} \in ||x_8||$ can be fired in $\N^s$. However, in $\N^c$ the marking of the place $p_{x7}$ labelled as $x_7$ is equal to 0, so transition $t_4 \in ||x_7||$ cannot be fired. Let us assume that $t_6$ is fired in $\N^s$.
\item The firing of $t_6$ in $\N^s$ does not imply changes in $\N^c$ because there exists no transition labelled $t_6$ in this net being a transition belonging to the support of a local T-semiflow but is not the first or the last one.    
\end {itemize}          
\end{example}

\ignore{
In order to reduce the restrictiveness of the approach, a relaxation is achieved in the control PN: the production of tokens from the sequences $t_a\rightarrow p \rightarrow t_b$ to the output buffers is advanced from transition $t_b$ to transition $t_a$. In this way, from the moment when the sequence $t_a\rightarrow p \rightarrow t_b$ start firing, the tokens are produced in the output buffers. These tokens allow that other T-semiflows can start its firing despite the fact that in the DSSP net the input buffers do not have enough tokens yet. 

Condition \ref{condi2} of the proposed control policy imposes that when a local T-semiflow start its firing in $\N^s$, all its transition must be fired until the end. In this way, the advance of the token production in the buffers of $\N^c$ do not compromise the liveness of the controlled system due to it is guaranteed (by condition 2) that the same token production will be made in the corresponding buffers of $\N^s$. }

\ignore{
\section{Composition of SSP system and control PN system}
Our approach is based on obtaining a control PN $\N^c$ that through guard expressions limit the firing of transitions in the SSP net $\N^s$. In this way, $\N^c$ can be seen as a high level scheduler that guides the firing of transitions in $\N^s$. Having a control PN is a great advantage in the case of complex systems, since $\N^c$ gives an overview of what is happening in the system. However, working with smaller systems could be interesting to have a unique net including $\N^s$ and $\N^c$. The idea is to impose the same initial conditions (1 and 2) that have been imposed through the control policy. To achieve this, some modifications are made in $\N^s$. First, in order that each local T-semiflow has its own ``first'' transition, some transitions are duplicated. Then the input buffers of each local T-semiflow are pre-assinged to its first transitions. At this point, condition 1 is fulfilled. Finally, the second condition is imposed by new places that are included in order to ``guide'' the complete firing of a local T-semiflow when its first transition has been fired. The steps to obtain this unique net system are as follows.
\begin{enumerate}
\item \emph{Duplicate first transitions}. Each first transition belonging to more than one local T-semiflow is duplicated as many times as the number of local T-semiflows to which belongs.
\item \emph{Pre-assignment of the buffers}. The input buffers of each local T-semiflow are pre-assinged to its first transition.
\item \emph{Include new places}. For each conflict transition $t_i$ (which is not a first transition), a place $p_{t_i}$ is added such that $p_{t_i}$ limits the firing of $t_i$ ($\b {Pre}[p_{t_i},t_i]=1$). Moreover, from each first transition ($t^k_j$) of the local T-semiflows to which $t_i$ belongs, there is an arc to $p_{t_i}$ ($\b {Post}[p_{t_i},t^k_j]=1$). 
\item \emph{Obtain the live control PN and ensure or force its liveness}. It is necessary to know if control buffers are needed to force the liveness in the control PN.
\item \emph{Include the control buffers}. If control buffers have been introduced to force the liveness of the control PN, these buffers must be added.
\end{enumerate}
}

\ignore{
Fig. \ref{DSSPunique} shows the equivalent controlled system $\N^e$ obtained from applying Step 1-5 to the non live SSP system $\N^s$ in Fig. \ref{fig:DSSPpb1}. The modifications performed are showed in different colors. In step 1 (blue) $t_1$ and $t_5$ have been transformed in two pairs of transition: \{$t^4_1$,$t^5_1$\} and \{$t^7_5$,$t^8_5$\}. In step 2 (green), the input buffers of the local T-semiflows has been preassigned to its first transitions. For example, $b_2$ is an input buffer of $\b x_7$ pre-assinged to its first transition $t^7_5$. Finally, in step 3 (red) for each transition in conflict relation that is not a first transition (\{$t_2$,$t_8$\} and \{$t_4$,$t_6$\}), a new place ($p_{t_2}$, $p_{t_8}$, $p_{t_4}$ and $p_{t_6}$) has been added. %Each one of these places ($p_{t_i}$) constrains the firing of its associated transition ($t_i$) such that $\b {Pre}[p_{t_i},t_i]=1$. Moreover, the marking of these places $p_{t_i}$ increases one when the first transition $t^k_j$ of a local T-semiflow to which its associated transition belongs $t_i$ is fired ($\b {Post}[p_{t_i},t^k_j]=1$). For example, as transition $t_2$ belong to $x_4$ there is an arc from $t^4_1$ to $p_{t_2}$. 
Step 5 is not applied because in Step 4 we obtain a live control PN (Fig. \ref{fig:control1}) without adding control buffers.

In the equivalent controlled net $\N^e$ of Fig. \ref{DSSPunique}, a local T-semiflow only can start if its input buffers have enough tokens to fire all its transitions. Moreover, once a local T-semiflow has started, all its transitions are fired. For example, $\b x_8$ can only start when its input buffer $b_3$ has a token because $b_3$ has been pre-assigned to transition $t^8_5$ . Once $\b x_8$ has started its firing by transition $t^8_5$, a token is produced in $p_{t_6}$. This token will enable the conflict transition $t_6$ and will guide the agent to fire the local T-semiflow $\b x_8$ completely.}

\begin{figure}[ht]
   \begin{center}
	\psfrag{t1}{$t_1$}\psfrag{t2}{$t_2$}\psfrag{t3}{$t_3$}\psfrag{t4}{$t_4$}
\psfrag{t5}{$t_5$}\psfrag{t6}{$t_6$}\psfrag{t7}{$t_7$}\psfrag{t8}{$t_8$}
\psfrag{t9}{$t_9$}\psfrag{t10}{$t_{10}$}\psfrag{t11}{$t_{11}$}\psfrag{t12}{$t_{12}$}
\psfrag{t13}{$t_{13}$}\psfrag{t14}{\blue{$t^4_1$}}\psfrag{t15}{\blue{$t^5_1$}}
\psfrag{p1}{$p_1$}\psfrag{p2}{$p_2$}\psfrag{pt2}{\red$p_{t_2}$}\psfrag{pt8}{\red$p_{t_8}$}\psfrag{pt4}{\red$p_{t_4}$}\psfrag{pt6}{\red$p_{t_6}$}
\psfrag{t54}{\blue{$t^7_5$}}\psfrag{t55}{\blue{$t^8_5$}}
\psfrag{p3}{$p_3$}\psfrag{p4}{$p_4$}
\psfrag{p5}{$p_5$}\psfrag{p6}{$p_6$}
\psfrag{p7}{$p_7$}\psfrag{p8}{$p_8$}
\psfrag{p9}{$p_9$}\psfrag{p10}{$p_{10}$}
\psfrag{p11}{$p_{11}$}
\psfrag{b1}{$b_1$}\psfrag{b2}{$b_2$}
\psfrag{b3}{$b_3$}\psfrag{b4}{$b_4$}\psfrag{b5}{$b_5$}
\psfrag{N1}{$\N_1$}\psfrag{N2}{$\N_2$}
    \includegraphics[width=.8\columnwidth]{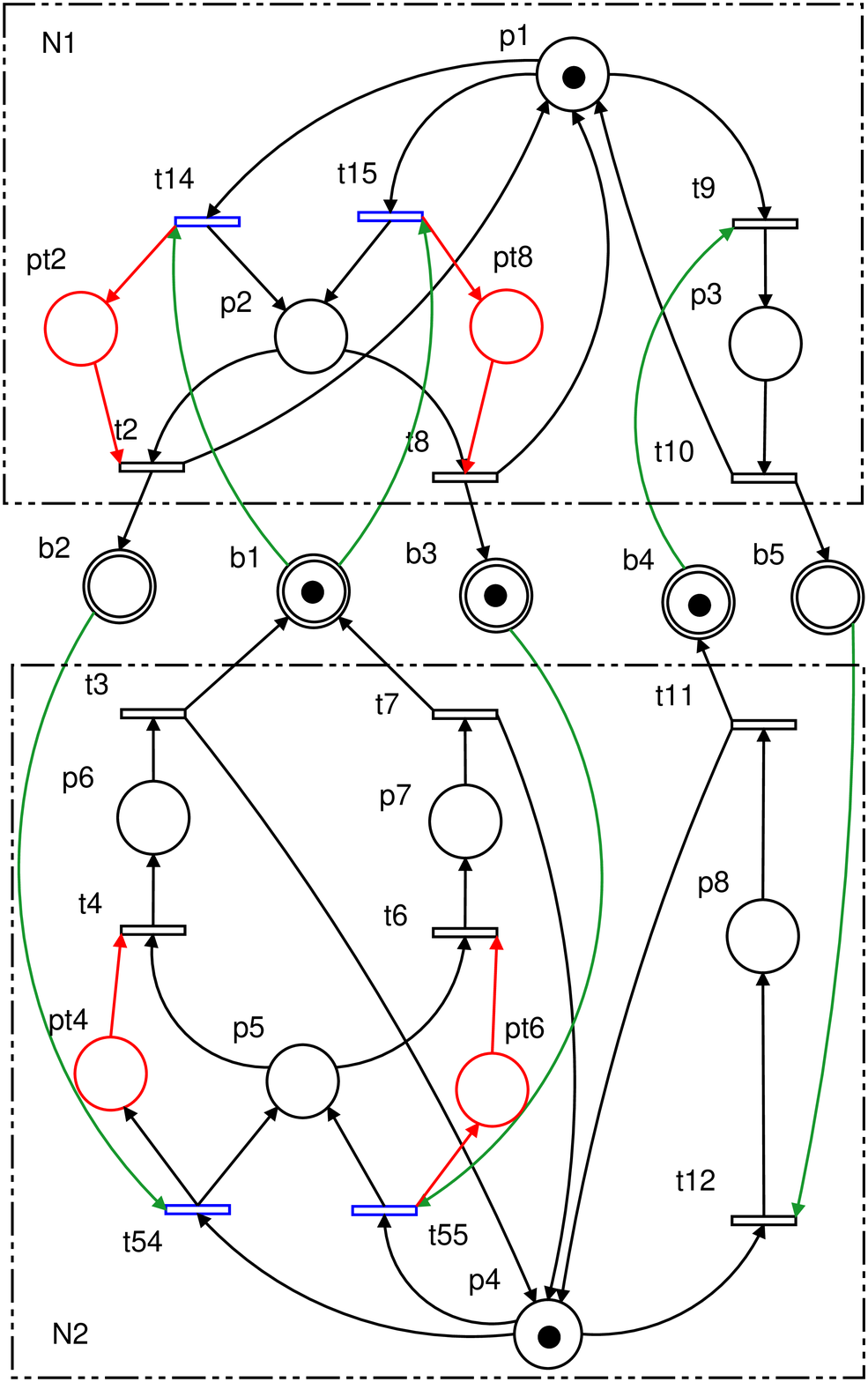}
\caption{PN system obtained by doing the synchronous composition of the SSP in Fig. \ref{fig:DSSPpb1} with its control PN in Fig. \ref{fig:control1}.}\label{DSSPunique}
\end{center}
\end{figure}

\ignore{
\begin{remark}
This paper presents the first method of enforcing liveness of SSP systems that keeps the distributed property of the net. As we mentioned before, the siphon-based method cannot be applyed in general, not only because the SSP will not be distributed but also because the SSP arte not ordinary. Moreover, the approach from RAS are also not aplicable in general since the class of systems are not comparable, as discussed in Section \ref{siphonbased}. However, the main drawback of the approach is the permissimibility of the approach as it can be seen in Tab. \ref{tab:simul} where are shown the number of reachable markings of the original not live SPP in Fig. \ref{fig:DSSPpb1}, the number of reachable markings after applying the siphon based method and the reachable marking of the net obtained after applying the approach in this paper (Fig. \ref{DSSPunique}).
\end{remark}}

\section{Conclusions}\label{sec:con}

Synchronized sequential processes (SSP) are modular PN systems used for modeling and analysis of systems composed by distributed cooperating sequential processes. This paper presents a liveness enforcement strategy for SSP systems formalized on two levels: \emph{execution} and \emph{control}. In the execution level the original SSP system evolves conditioned by the control level, a kind scheduler: each transition in SSP net has associated a guard expression which depends on the state of the control level. A \emph{control PN} is obtained from the SSP structure by using Alg. \ref{alg:2}. Control PN models the consumption/production relation between buffers and local T-semiflows of the SSP net. If the control PN is not live, Alg. \ref{conpla} enforce its liveness by adding some control places. Both SSP and control PN evolve synchronously following a given control evolution policy synthesized by Alg. \ref{policy}. 

To control is always to constraint the possible behaviors of the plant. Since the method presented in this paper is constraining the firing of local T-semiflows of the global T-semiflows (not allowing to start firing again the global T-semiflow until all local T-semiflows are not fired in the corresponding proportions) the permissitivity of this approach may be less than of other approaches in literature. For example, by doing the synchronous composition of the non live SSP system in Fig. \ref{fig:DSSPpb1} with its supervisory control PN system in Fig. \ref{fig:control1}, the net system in Fig. \ref{DSSPunique} is obtained. 

\begin{table}[htb]
\centering
\caption{Simulation result using SSP net system in Fig. \ref{fig:DSSPpb1} }
\label{tab:simul}
\begin{tabular}{|c|c|c|}
\hline
\textbf{Net system}   & \textbf{\# Reachable makings} & \textbf{\# livelock marking} \\ \hline
SSP (Fig. \ref{fig:DSSPpb1})           & 180                 & 13                 \\ \hline
SSP + monitors   & 139                 & 0                  \\ \hline
SSP + control (Fig. \ref{DSSPunique}) & 94                  & 0                  \\ \hline
\end{tabular}
\end{table}

Tab. \ref{tab:simul} shows the number of reachable markings of the non live SSP in Fig. \ref{fig:DSSPpb1}, the number reachable markings of the live net system obtained by controlling the bad siphons (method well known understood in RAS) and the number of reachable markings of the composition of the SSP with the control net proposed in this paper. The number of reachable markings is in general greater if the bad siphons are controlled. Nevertheless, as it is discussed in Section \ref{sec:dssp}, unfortunately distributeveness of the system is lost. Up to our knowledge, the approach in this paper is one of the first one on dealing with liveness enforcement in this class, method that keeps the controlled system to be also a SSP (hence having this distributiveness property).

As a future work, the permissibility of the approach with respect to the number of reachable marking will be improved. By using the Synchrony Theory \cite{TRPetri76,ICSi87} the liveness of the control PN could be enforced and in some cases more permissive approach as the one presented in Sec. \ref{sec:livefor} could be obtained. Another future work will be to check permissibility with respect to the throughput. Notice that timed SSP could be throughput non monotone systems (contrary to DSSP systems) and a reduction in the number of reachable markings is not necessary always worst or at least in not the same proportion as the reduction in the number of reachable markings.

\bibliography{2019_TAC}

\begin{thebibliography}{10}
\providecommand{\url}[1]{#1}
\csname url@samestyle\endcsname
\providecommand{\newblock}{\relax}
\providecommand{\bibinfo}[2]{#2}
\providecommand{\BIBentrySTDinterwordspacing}{\spaceskip=0pt\relax}
\providecommand{\BIBentryALTinterwordstretchfactor}{4}
\providecommand{\BIBentryALTinterwordspacing}{\spaceskip=\fontdimen2\font plus
\BIBentryALTinterwordstretchfactor\fontdimen3\font minus
  \fontdimen4\font\relax}
\providecommand{\BIBforeignlanguage}[2]{{%
\expandafter\ifx\csname l@#1\endcsname\relax
\typeout{** WARNING: IEEEtran.bst: No hyphenation pattern has been}%
\typeout{** loaded for the language `#1'. Using the pattern for}%
\typeout{** the default language instead.}%
\else
\language=\csname l@#1\endcsname
\fi
#2}}
\providecommand{\BIBdecl}{\relax}
\BIBdecl

\bibitem{ARMura89}
T.~Murata, ``{Petri} nets: Properties, analysis and applications,''
  \emph{Proceedings of the IEEE}, vol.~77, no.~4, pp. 541--580, 1989.

\bibitem{ArSiTeCo98}
M.~Silva, E.~Teruel, and J.-M. Colom, ``Linear algebraic and linear programming
  techniques for the analysis of {P/T} net systems.'' \emph{Lecture on Petri
  Nets I: Basic Models}, vol. 1491, pp. 309--373, 1998.

\bibitem{ezpeleta}
J.~{Ezpeleta}, J.~M. {Colom}, and J.~{Martinez}, ``{A Petri net based deadlock
  prevention policy for flexible manufacturing systems},'' \emph{IEEE Trans. on
  Robotics and Automation}, vol.~11, no.~2, pp. 173--184, April 1995.

\bibitem{BOIoAn06}
M.-V. Iordache and P.-J. Antsaklis, \emph{Supervisory Control of Concurrent
  Systems: A Petri Net Structural Approach}.\hskip 1em plus 0.5em minus
  0.4em\relax Birkhauser Boston, 2006.

\bibitem{SilvaDSSP}
L.~Recalde, E.~Teruel, and M.~Silva, ``Modeling and analysis of sequential
  processes that cooperate through buffers,'' \emph{IEEE Trans. Robotics and
  Automation}, vol.~14, no.~2, pp. 267 -- 277, 1998.

\bibitem{park2001deadlock}
J.~Park and S.~Reveliotis, ``Deadlock avoidance in sequential resource
  allocation systems with multiple resource acquisitions and flexible
  routings,'' \emph{IEEE Transactions on Automatic Control}, vol.~46, no.~10,
  pp. 1572--1583, Oct 2001.

\bibitem{Colom:2003}
J.~M. Colom, ``The resource allocation problem in flexible manufacturing
  systems,'' in \emph{Applications and {T}heory of {P}etri {N}ets}, ser.
  Lecture Notes in Computer Science, W.~Aalst and E.~Best, Eds.\hskip 1em plus
  0.5em minus 0.4em\relax Berlin, Heidelberg: Springer-Verlag, 2003, vol. 2679,
  pp. 23--35.

\bibitem{li2004elementary}
Z.~Li and M.~Zhou, ``Elementary siphons of {P}etri nets and their application
  to deadlock prevention in flexible manufacturing systems,'' \emph{IEEE
  Transactions on Systems, Man and Cybernetics, Part A: Systems and Humans},
  vol.~34, no.~1, pp. 38--51, Jan 2004.

\bibitem{cano2012}
E.~Cano, C.~Rovetto, and J.~M. Colom, ``\BIBforeignlanguage{English}{An
  algorithm to compute the minimal siphons in {S$^4$PR} nets},''
  \emph{\BIBforeignlanguage{English}{Discrete Event Dynamic Systems}}, vol.~22,
  no.~4, pp. 403--428, 2012.

\bibitem{7870672}
S.~{Wang}, D.~{You}, and M.~{Zhou}, ``{A Necessary and Sufficient Condition for
  a Resource Subset to Generate a Strict Minimal Siphon in S4PR},'' \emph{{IEEE
  Trans. on Automatic Control}}, vol.~62, no.~8, pp. 4173--4179, 2017.

\bibitem{ARBeMaAl19}
S.~Bernardi, C.~Mahulea, and J.~Albareda, ``Toward a decision support system
  for the clinical pathways assessment,'' \emph{Discrete Event Dynamic Systems:
  Theory and Applications}, vol.~29, no.~1, pp. 91--125, 2019.

\bibitem{clavelCDC}
D.~{Clavel}, C.~{Mahulea}, and M.~{Silva}, ``{On liveness enforcement of DSSP
  net systems},'' in \emph{{2016 IEEE 55th Conference on Decision and Control
  (CDC)}}, Dec 2016, pp. 3935--3941.

\bibitem{ICSilv93b}
M.~Silva, ``Introducing {Petri} nets,'' in \emph{Practice of {Petri} Nets in
  Manufacturing}.\hskip 1em plus 0.5em minus 0.4em\relax Chapman \& Hall, 1993,
  pp. 1--62.

\bibitem{GMEC}
A.~Giua, F.~DiCesare, and M.~Silva, ``Generalized mutual exclusion contraints
  on nets with uncontrollable transitions,'' in \emph{IEEE Int. Conf. on
  Systems, Man and Cybernetics}.\hskip 1em plus 0.5em minus 0.4em\relax IEEE,
  1992, pp. 974--979.

\bibitem{ARLuWuZhSh20}
J.~{Luo}, W.~{Wu}, M.~{Zhou}, H.~{Shao}, K.~{Nonami}, and H.~{Su}, ``Structural
  controller for logical expression of linear constraints on petri nets,''
  \emph{IEEE Transactions on Automatic Control}, vol.~65, no.~1, pp. 397--403,
  2020.

\bibitem{IPTrGaCoEz05}
F.~Tricas, F.~Garc{\'\i}a-Vall{\'e}s, J.~Colom, and J.~Ezpeleta, ``{A Petri Net
  Structure-Based Deadlock Prevention Solution for Sequential Resource
  Allocation Systems},'' in \emph{2005 IEEE International Conference on
  Robotics \& Automation (ICRA{\textquoteright}05)}, Barcelona, Spain, 2005,
  pp. 272--278.

\bibitem{ARTeCoSi97}
E.~{Teruel}, J.~M. {Colom}, and M.~{Silva}, ``{Choice-free Petri nets: a model
  for deterministic concurrent systems with bulk services and arrivals},''
  \emph{IEEE Transactions on Systems, Man, and Cybernetics - Part A: Systems
  and Humans}, vol.~27, no.~1, pp. 73--83, 1997.

\bibitem{TRPetri76}
C.-A. Petri, ``Interpretations of net theory,'' Technical Report 75-07.
  Gessellschaft fur Mathematik und Datenverarbeitung, Tech. Rep., 1976.

\bibitem{ICSi87}
M.~Silva, ``{Towards a Synchrony Theory for P/T Nets. Concurrency and Nets},''
  in \emph{Advances in Petri Nets}.\hskip 1em plus 0.5em minus 0.4em\relax
  Springer-Verlag, 1987, pp. 435--460.

\end{thebibliography}

\end{document}